\documentclass[11pt,tbtags,reqno]{article}
\usepackage[margin=0.96in, top=2.2cm,bottom=2.2cm]{geometry}

\usepackage[english]{babel}
\usepackage[T1]{fontenc}
\usepackage{setspace}
\usepackage{enumitem}
\usepackage[normalem]{ulem}
\usepackage{float}
\usepackage{scrextend}
\deffootnote{0em}{1.6em}{\thefootnotemark.\enskip}
\usepackage{xcolor}
\usepackage{latexsym}
\usepackage{graphicx}
\usepackage{tikz}
\usepackage{csquotes}

\usepackage{amsmath}
\usepackage[makeroom]{cancel}
\usepackage{amssymb}
\usepackage{amsthm}
\usepackage{mathtools}
\usepackage{euscript,mathrsfs}
\usepackage{bbm}
\usepackage[colorinlistoftodos]{todonotes}
\PassOptionsToPackage{hyphens}{url}
\usepackage[colorlinks=true, allcolors=blue]{hyperref}
\usepackage{titlesec}
\usepackage[title]{appendix}
\titlelabel{\thetitle.\quad}

\mathtoolsset{showonlyrefs}

\usepackage[
    citestyle=numeric,
    bibstyle=authoryear,
    dashed=false,
    maxnames=4]{biblatex}
\makeatletter
\input{numeric.bbx}
\makeatother
\renewbibmacro{in:}{}
\DeclareNameAlias{author}{family-given}

\DeclareFieldFormat*{title}{#1}
\DeclareFieldFormat*[book]{title}{\textit{#1}}
\DeclareFieldFormat*{volume}{\textbf{#1}}

\hypersetup{
    colorlinks,
    linkcolor={blue},
    citecolor={blue},
    urlcolor={blue}
}
\usepackage{yfonts}
\usepackage{indentfirst}
\usepackage[parfill]{parskip}
\newtheorem{Th}{Theorem}[section]
\newtheorem{Def}[Th]{Definition}
\newtheorem{Lem}[Th]{Lemma}
\newtheorem{Prop}[Th]{Proposition}
\newtheorem{Cor}[Th]{Corollary}
\theoremstyle{definition}
\newtheorem{remark}[Th]{Remark}

\numberwithin{equation}{section}

\newcommand{\R}{\mathbb{R}}

\newcommand{\eps}{\varepsilon}

\title{\textsc{The Game of Marginal Utilities}}
\author{
\textsc{Isaac M. Sonin}
\thanks{\textsc{Department of Mathematics and Statistics, University of North Carolina at Charlotte, Charlotte, NC 28223, USA} (e-mail: {\it imsonin@uncc.edu})}
\and
\textsc{Georgy Gaitsgori}
\thanks{\textsc{Department of Mathematics, Columbia University, 2990 Broadway, New York, NY 10027, USA} (e-mail: {\it gg2793@columbia.edu})}
\and
\textsc{Yaakov Malinovsky}
\thanks{\textsc{Department of Mathematics and Statistics, University of Maryland, Baltimore County, Baltimore, MD 21250, USA} (e-mail: {\it yaakovm@umbc.edu}). Research supported in part by BSF grant 2020063.}
}
\date{\today}

\begin{document}
\maketitle

\begin{abstract}
We study a noncooperative resource-allocation game in which $m$ players distribute fixed resources among $n$ projects and the payoff of player $j$ is given by
\[
F^j(x)
=
\sum_{i=1}^n
\frac{a_i x_i^j}{b_i+\sum_{\ell=1}^m x_i^\ell},
\]
where $a_i$ and $b_i$ are project parameters, while $x_i^j$ is the amount of resources player $j$ allocates to project $i$.
This specification combines diminishing returns with congestion generated by competitors. We prove that the game has a unique Nash equilibrium and characterize it by an equimarginal principle. We show that, after the projects are ordered by $a_i/b_i$, each player invests in an initial segment of projects, and these segments are nested across players, so the equilibrium decomposes into consecutive activity zones; players with larger resources invest weakly more in every project. In the fully active regime, where every player invests in every project, we reduce the equilibrium to a single scalar nonlinear equation for the aggregate marginal-utility rate; all individual marginal rates, investments, and payoffs then follow from explicit formulas. 
We also provide a projected marginal-utility algorithm with global linear convergence under an explicit step-size condition, together with a structure-exploiting \textit{Block Pandora} algorithm that reconstructs and certifies the equilibrium conditional on a proposed nested cutoff structure.
\end{abstract}



\section{Introduction and Motivation}\label{section_introduction}

How should strategic decision makers divide fixed resources among a common collection of projects when additional investment raises the return from a project, but investment by competitors dilutes that return? 
At a deliberately macroscopic scale, one may imagine a small number of large investors choosing among a small number of global projects, for example, artificial intelligence, energy, biotechnology, or real estate; somewhat tongue-in-cheek, the model may be viewed as an ``investment problem for trillionaires.''
More formally, we study the following static noncooperative model. There are $m$ players and $n$ projects. Player $j$ has a fixed resource $r^j>0$ and chooses nonnegative allocations $x_i^j$ satisfying $\sum_i x_i^j=r^j$. Her payoff is
\begin{equation}\label{eq_introduction_payoff}
    F^j(x)
    =
    \sum_{i=1}^n
    \frac{a_i x_i^j}
         {b_i+\sum_{\ell=1}^m x_i^\ell},
\end{equation}
where $a_i>0$ is the return scale of project $i$ and $b_i>0$ is a baseline load. The baseline may represent capital already committed by agents outside the game, exogenous congestion, or an initial stock against which the players' investments are diluted. For fixed opponents, each summand is increasing and strictly concave in the player's own allocation. At the same time, an increase in a competitor's allocation enlarges the denominator and reduces the payoff of the player under consideration.
Thus, \eqref{eq_introduction_payoff} combines diminishing returns with strategic crowding in a particularly simple rational form.

The same payoff has a literal dilution interpretation. Think of project $i$ as a vessel containing an amount $a_i$ of a valuable substance in an initial volume $b_i$. Player $j$ adds a volume $x_i^j$ and later withdraws the same volume; her recovered amount of the substance is exactly the $i$th term in \eqref{eq_introduction_payoff}. The one-player version of this ``water puzzle'' was studied by Malinovsky and Sonin~\cite{SoninMalinovsky}, who obtained explicit cutoff and allocation formulas. The present paper develops the strategic multi-player counterpart, and our formulas reduce to theirs when $m=1$. A related water-filling analogy from information theory is discussed later in the introduction.

The natural optimality condition in our setting is the equimarginal principle: at equilibrium, every project used by a given player has the same marginal utility for that player, while every unused project has no larger marginal utility. These are precisely the Karush--Kuhn--Tucker conditions for the player's concave best-response problem. The main difficulty is that the active constraints are endogenous. The projects used by each player are not known in advance, and changing one allocation changes the marginal utilities of all competitors on the same project. A direct KKT formulation therefore combines a nonlinear continuous system with an active-set problem.

\paragraph{Main results.}
The results of the paper have three layers. 
The first layer places the model within the classical theory of concave games initiated by Rosen~\cite{Rosen}. Carrying out this step still requires model-specific verification, most importantly, checking Rosen's diagonal strict concavity condition, but the underlying existence and uniqueness results come from established theory, while the equimarginal characterization follows from the KKT conditions. The main contributions of the paper lie instead in the second and third layers: the detailed structural and essentially explicit solution of the equilibrium, and the construction of one globally convergent algorithm and one structure-exploiting reconstruction and certification method.

Our principal structural result concerns the equilibrium supports. Ordering the projects by their quality ratios
\(
    q_i:=a_i/b_i, \,
    q_1\ge\cdots\ge q_n,
\)
and the players by their resources,
\(
    r^1\ge\cdots\ge r^m,
\)
we show that each player $j$ invests in an initial segment $\{1,\ldots,k^j\}$ of the project list and that the cutoffs are nested:
\(
    n\ge k^1\ge\cdots\ge k^m\ge1.
\)
Consequently, the ordered projects decompose into consecutive activity zones on which the set of active players is constant. Players with larger resources have weakly smaller equilibrium marginal rates, are active in weakly more projects, and invest weakly more in every project. If, in addition, $b_1\ge\cdots\ge b_n$, then each player's equilibrium allocation is also weakly decreasing across projects.

The equilibrium is especially explicit in the fully active regime, where every player invests in every project. Let
\[
    R:=\sum_{j=1}^m r^j,
    \qquad
    B:=\sum_{i=1}^n b_i,
    \qquad
    C:=\sum_{j=1}^m c^j,
\]
where $c^j$ is the common equilibrium marginal utility rate of player $j$. We show that $C$ is the unique positive solution of the scalar equation
\begin{equation}\label{eq_introduction_basic_equation}
    R+B
    =
    \sum_{i=1}^n \frac{a_i}{2C}
    \left(
        (m-1)+
        \sqrt{(m-1)^2+\frac{4b_iC}{a_i}}
    \right).
\end{equation}
Once $C$ is known, all project loads, individual marginal rates, allocations, and payoffs follow from explicit formulas. In particular, the first step of the construction depends on the resource vector $(r^1,\ldots,r^m)$ only through its sum $R$; the individual resources enter only in the subsequent affine reconstruction. The same formulas also decide whether the fully active regime is valid, see Section \ref{section_local_structure} for more details.
To the best of our knowledge, the aggregate scalar reduction, the accompanying recovery formulas, and this support criterion have not appeared in the related multi-project allocation literature. 

Outside the fully active regime, the formulas within each activity zone remain explicit, but the zones themselves and the division of each player's resources among them must be identified. We address this remaining problem by giving two complementary constructive methods. The projected marginal-utility iteration requires no knowledge of the active sets: every player moves simultaneously in the direction of her marginal utilities and projects back onto her resource simplex. Under an explicit step-size condition, the resulting map is a contraction and converges globally at a linear rate. We also introduce the structure-exploiting \textit{Block Pandora} algorithm. Conditional on a proposed nested cutoff pattern, it recursively glues fully active zones through one-dimensional marginal-rate matching equations and reconstructs the unique strict equilibrium of the corresponding restricted game whenever a strict root is bracketed. The omitted-project marginal inequalities then give an exact answer as to whether the reconstructed profile is the Nash equilibrium of the original game. Thus, the projected method is unconditional, whereas Block Pandora is a support-based reconstruction and certification device that exposes how the equilibrium is assembled from its zones.

\paragraph{The role of the baseline and related literature.}
Having described the main results and algorithms, we now place the model in context. The cleanest comparison is obtained by removing the positive baselines, which both isolates their structural role and connects the game to the classical lottery-Blotto and proportional-allocation literature. Formally setting $b_i=0$ gives
\begin{equation}\label{eq_introduction_zero_baseline}
    \widehat F^j(x)
    =
    \sum_{i=1}^n
    a_i\frac{x_i^j}{\sum_{\ell=1}^m x_i^\ell}.
\end{equation}
For $m\ge2$, under the usual convention that an unclaimed project yields zero payoff, this game has the unique Nash equilibrium
\begin{equation}\label{eq_introduction_zero_baseline_solution}
    x_i^{j,\ast}
    =
    r^j\frac{a_i}{\sum_{h=1}^n a_h},
    \qquad i=1,\ldots,n,
    \quad j=1,\ldots,m.
\end{equation}
This proportional equilibrium is classical in the lottery-Blotto literature. Friedman~\cite{Friedman} obtained the two-player result in an advertising-allocation model, and Duffy and Matros~\cite{DuffyMatros} extended it to an arbitrary finite number of players in the common-value setting. Proposition~\ref{prop_zero_baseline_equilibrium} in Appendix~\ref{appendix_zero_baseline_equilibrium} provides a short direct proof for completeness and records the benchmark against which the role of the positive baselines can be seen. Under \eqref{eq_introduction_zero_baseline_solution}, every player uses every project and all players hold the same portfolio proportions. By contrast, positive $b_i$ can create zero allocations, ordered cutoffs, and several activity zones. The baseline therefore changes the qualitative geometry of equilibrium rather than merely regularizing the payoff at zero. Section~\ref{section_examples} illustrates this transition by scaling a fixed baseline vector continuously from zero.

The baseline-free ratio form has also been studied under richer valuation structures. Kim, Kim, and Kim~\cite{KimEtAl} consider a lottery-Blotto game in which the value of each project may depend on the player. They establish equilibrium existence for an arbitrary finite number of players and characterize all Nash equilibria in the two-player case. Viewed as a market mechanism, \eqref{eq_introduction_zero_baseline} is also a common-value linear-utility instance of the Trading Post, or Shapley--Shubik, mechanism. Br\^anzei, Gkatzelis, and Mehta~\cite{BranzeiGkatzelisMehta} study Trading Post for general concave utilities, with an emphasis on Nash social welfare and individual fairness rather than on an explicit characterization of equilibrium supports. These models permit substantially richer valuations or utility functions, but retain the zero-baseline proportional-allocation rule.

The closest positive-baseline specification that we are aware of is the continuous Chinese auction with an auctioneer ticket in each basket. Br\^anzei, Forero, Larson, and Miltersen~\cite{BranzeiEtAl} consider the more general payoff $\sum_{i=1}^n v_i^j x_i^j/(\Delta_i+\sum_{\ell=1}^m x_i^\ell)$, where the valuations $v_i^j$ may depend on the player and $\Delta_i>0$ is a ticket placed by the auctioneer.
They show, in particular, that positive auctioneer tickets guarantee the existence of a pure Nash equilibrium in the continuous game. Our model is the common-value specialization $v_i^j=a_i$ and $\Delta_i=b_i$. Although this specialization is narrower with respect to valuations, it permits substantially stronger conclusions: uniqueness for an arbitrary number of players, a nested cutoff and activity-zone structure, explicit formulas within each fully active zone, the aggregate scalar reduction \eqref{eq_introduction_basic_equation}, and the constructive methods developed in Section~\ref{section_algorithms}.

At the level of the general equilibrium framework, our starting point is Rosen's theory of concave games~\cite{Rosen}. The verification of diagonal strict concavity is specific to the present payoff, but the underlying existence and uniqueness framework is classical. The projectwise ratio form also connects the model to the linear Tullock contest success function and to rent-seeking contests~\cite{CornesHartley,Skaperdas,SzidarovszkyOkuguchi,Tullock}. Related branches of the literature include multi-prize contests~\cite{ClarkRiis}, static and dynamic multi-battle Blotto games~\cite{DuffyMatros,KimEtAl,AnbarciEtAl}, and strategic resource allocation under congestion. In the latter direction, Johari and Tsitsiklis~\cite{JohariTsitsiklis} study pricing and efficiency loss in a network resource-allocation game, whereas our emphasis is on the exact geometry and computation of the Nash equilibrium.

There is also a useful analogy with water-filling methods in information theory. In a one-agent allocation problem, and in each best-response problem here, optimality equalizes marginal returns across the active channels. Iterative water-filling algorithms use this principle to allocate power across communication channels, including in multi-user Gaussian systems~\cite{YuEtAl}. The analogy is structural rather than literal: Shannon-type channel utilities are typically logarithmic, whereas our payoff is the rational function in \eqref{eq_introduction_payoff}, and none of our equilibrium formulas follows from the information-theoretic results. Finally, the projected marginal-utility method is naturally expressed in the language of monotone variational inequalities and projection algorithms; see Facchinei and Pang~\cite{FacchineiPang} for the general framework. The explicit strong anti-monotonicity and Lipschitz estimates used to obtain a global linear rate are specific to the present model.

The rest of the paper is organized as follows. Section~\ref{section_model} introduces the model and states the main theorem. Section~\ref{section_existence_equimarginal} proves existence, uniqueness, and the equimarginal characterization. Section~\ref{section_local_structure} derives the fully active solution, the scalar equation, and the full-activity criterion. Section~\ref{section_global_structure} establishes the nested cutoff, zone, and monotonicity structure of the general equilibrium. Section~\ref{section_algorithms} presents the projected marginal-utility algorithm and the Block Pandora reconstruction algorithm. Section~\ref{section_examples} gives fully active, baseline-continuation, multi-zone, and best-response cycling examples. Section~\ref{section_open_problems} discusses extensions, welfare, and
mechanism-design questions.

\section{Model and Main Result}\label{section_model}

We consider a noncooperative game with $m$ players and $n$ projects. The lower index $i$ always denotes a project, $i=1,\ldots,n$, and the upper index $j$ always denotes a player, $j=1,\ldots,m$. Moreover, the symbol $^\ast$ will always denote an equilibrium quantity.
In particular, $x_i^j$ denotes the amount invested by player $j$ in project $i$, whereas $x_i^{j,\ast}$ denotes this amount at equilibrium. We highlight that, throughout the paper, all resources and all project parameters are positive unless stated otherwise.

Each project $i$ is described by two parameters: a \textit{return rate} parameter $a_i>0$ and a \textit{congestion} or \textit{penalization} parameter $b_i>0$; we write $a=(a_1,\ldots,a_n), \, b=(b_1,\ldots,b_n)$. Intuitively, one may think of the projects as ranked by their primitive return rates, so that larger values of $a_i$ correspond to more attractive projects. When $b_i$ is interpreted as background or market-wide investment, it is also natural in many applications to expect the more attractive projects to carry larger values of $b_i$. Thus, one may keep in mind the benchmark case in which $a$, $b$, and the quality ratios $q_i = a_i / b_i$ are all similarly ordered. These three orderings need not be compatible for arbitrary parameter vectors, however, and our results do not require separate orderings of $a$ and $b$.
The ordering relevant for the equilibrium structure is determined by the quality parameter
\[
    q_i:=\frac{a_i}{b_i}.
\]
Accordingly, we assume, without loss of generality, that
\(
    q_1\ge q_2\ge\cdots\ge q_n>0.
\)
This is the only standing ordering assumption on the projects; any additional ordering assumption, such as monotonicity of $b$, will be stated explicitly.

A strategy of player $j$ is a vector of resource allocations
\(
    x^j=(x_1^j,\ldots,x_n^j)\in X^j,
\)
where
\begin{equation}\label{def_X_j}
    X^j=\left\{x^j\in \R_+^n:\sum_{i=1}^n x_i^j=r^j\right\}.
\end{equation}
Here, $r^j > 0$ is the total resource of player $j$. Without loss of generality, we assume that $r^1 \ge r^2 \ge \dots \ge r^m$ and write $r = (r^1, \dots, r^m)$.
A strategy profile is denoted by
\begin{equation}\label{def_X}
    x=(x^1,\ldots,x^m)\in X:=X^1\times\cdots\times X^m.
\end{equation}
We call such a profile admissible. For an admissible profile $x$, we define the investment of all players except $j$ in project $i$ by
\[
    u_i^j=u_i^j(x):=\sum_{\ell\ne j}x_i^\ell,
\]
and the total \textit{load} of project $i$ by
\[
    L_i=L_i(x):=b_i+\sum_{\ell=1}^m x_i^\ell=b_i+u_i^j+x_i^j.
\]

The payoff of player $j$ from project $i$ is given by
\[
    f_i^j(x):=\frac{a_i x_i^j}{L_i(x)}
    =\frac{a_i x_i^j}{b_i+\sum_{\ell=1}^m x_i^\ell},
\]
and the total payoff of player $j$ is given by
\[
    F^j(x):=\sum_{i=1}^n f_i^j(x)
    =\sum_{i=1}^n \frac{a_i x_i^j}{b_i+\sum_{\ell=1}^m x_i^\ell}.
\]
Thus, a larger value of $a_i$ makes project $i$ more attractive, while a larger total load $L_i$ makes the same project more congested, equivalently less attractive.

For later use, we record the \textit{marginal utility} of player $j$ on project $i$:
\begin{equation}\label{marginal_utility_formula}
    c_i^j = c_i^j(x):=\frac{\partial F^j}{\partial x_i^j}(x)
    =\frac{\partial f_i^j}{\partial x_i^j}(x)
    =\frac{a_i\bigl(L_i(x)-x_i^j\bigr)}{L_i(x)^2}
    =\frac{a_i\bigl(b_i+u_i^j(x)\bigr)}{L_i(x)^2}.
\end{equation}
At this stage, $c_i^j$ is only a notation for the partial derivative. The equimarginal property, where these quantities become equal over the active projects of a fixed player, will be proved in the next section.
We denote the active set of player $j$ at a profile $x$ by
\(
    A^j(x):=\{i: x_i^j>0\}.
\)

Finally, we write $M(n,m,a,b,r)$ for the model and record the following standard definition for completeness of exposition.

\begin{Def}[Nash equilibrium]
An admissible profile $x^\ast\in X$ is a Nash equilibrium if, for every player $j=1,\ldots,m$,
\(
    F^j(x^{j,\ast},x^{-j,\ast})\ge F^j(y^j,x^{-j,\ast})
\)
for all $y^j\in X^j$.
Here, $x^{-j,\ast}$ denotes the strategies of all players except player $j$ at the profile $x^\ast$.
\end{Def}

We now summarize the main results in a compact form, whose precise versions are proved in the sections below.

\begin{Th}\label{thm_main_structure}
The game $M(n,m,a,b,r)$ has a unique Nash equilibrium $x^\ast$.  This equilibrium has the following structure.
\begin{enumerate}[label=\textup{(\roman*)}]
    \item \textup{(Equimarginal principle.)} For each player $j$, there is a common marginal utility rate $c^j$ on all projects used by that player; every unused project has marginal utility no larger than $c^j$.\label{thm_Equimarginal}

    \item \textup{(Ordered active projects.)} After projects are ordered by their quality parameters $q_i=a_i/b_i$ and players are ordered so that $r^1\ge r^2\ge\cdots\ge r^m$, each player invests only in the first $k^j$ projects for some $n \ge k^1 \ge \dots \ge k^m \ge 1$. \label{thm_Ordered}

    \item \textup{(One-zone formula.)}
    In the fully active regime, where $x_i^{j,\ast}>0$ for all $i,j$, the cumulative marginal utility rate
    \(
        C=\sum_{j=1}^m c^j
    \)
    is determined by one scalar nonlinear equation involving only the total resource
    \(
        R=\sum_{j=1}^m r^j
    \)
    and the project parameters. Once $C$ is known, all individual marginal rates, allocations, and payoffs are explicit.
    \label{thm_one_zone}

    \item \textup{(Monotonicity.)} The equilibrium is monotone across players: players with larger resources have weakly smaller marginal utility rates and invest weakly more in each project.  In the monotone case $b_1 \ge b_2 \ge \dots \ge b_n$, the allocations are also monotone across projects. \label{thm_monotonicity}

    \item \textup{(Constructive solution.)} The equilibrium can be computed by two complementary algorithms: a globally and linearly convergent projected marginal-utility method, and the model-specific \textit{Block Pandora} algorithm, which exploits the nested zone structure. \label{thm_algo}
\end{enumerate}
\end{Th}

Theorem~\ref{thm_main_structure} is proved in stages. Formally, it is the combination of Propositions~\ref{thm_existence_uniqueness}, \ref{thm_one_zone_formulas}, \ref{thm_cutoff_zone_structure}, \ref{thm_global_monotonicity}, \ref{thm_projected_algorithm_convergence}, and \ref{prop_block_pandora_correctness}. More precisely, Proposition~\ref{thm_existence_uniqueness} proves existence and uniqueness and gives part~\ref{thm_Equimarginal}; Proposition~\ref{thm_one_zone_formulas} gives part~\ref{thm_one_zone}; Propositions~\ref{thm_cutoff_zone_structure} and \ref{thm_global_monotonicity} give parts~\ref{thm_Ordered} and \ref{thm_monotonicity}, respectively; and Propositions~\ref{thm_projected_algorithm_convergence} and \ref{prop_block_pandora_correctness} give the two constructive statements in part~\ref{thm_algo}.

\section{Existence, Uniqueness, and Equimarginal Characterization}\label{section_existence_equimarginal}

In this section, we show that our game admits a unique Nash equilibrium and provide its \textit{equimarginal} characterization. 
Existence and uniqueness follow from Rosen's classical theory of concave games~\cite[Theorems~1 and~2, p.~524]{Rosen} once diagonal strict concavity is verified. The purpose of this section is therefore to check Rosen's hypotheses and record the equilibrium conditions used throughout the paper.

Following Rosen, we introduce the pseudo-gradient of the game. It is useful first to visualize its coordinates as the marginal-utility matrix
\[
    \begin{pmatrix}
        c_1^1(x) & \cdots & c_n^1(x)\\
        \vdots   & \ddots & \vdots\\
        c_1^m(x) & \cdots & c_n^m(x)
    \end{pmatrix},
\]
whose $(j,i)$-entry is the marginal utility $c_i^j(x)$ from \eqref{marginal_utility_formula}. Formally, however, all inner products and norms below are taken after stacking the rows of this matrix into a vector. Thus, we define
\begin{equation}\label{def_Phi}
    \Phi(x)
    :=
    \bigl(
        \nabla_{x^1}F^1(x),\ldots,
        \nabla_{x^m}F^m(x)
    \bigr)
    \in\R^{mn}.
\end{equation}
Equivalently, the coordinate of $\Phi(x)$ corresponding to player $j$
and project $i$ is
\(
    \Phi_i^j(x)=c_i^j(x).
\)
We record the needed verification in the next lemma. The proof is a direct calculation and is given in Appendix~\ref{appendix_lem_concavity_antimonotonicity}. The quantitative bound \eqref{eq_strong_antimonotonicity} will be a crucial ingredient in the proof of the algorithm convergence in Section \ref{section_algorithms} below.

\begin{Lem}[Concavity and diagonal strict concavity]\label{lem_concavity_antimonotonicity}
For every player $j$ and every fixed profile $x^{-j}$ of the other players, the map
\(
    x^j\mapsto F^j(x^j,x^{-j})
\)
is strictly concave on $X^j$. Moreover, letting
\begin{equation}\label{def_R_mu}
    R:=\sum_{j=1}^m r^j
    \qquad
    \text{ and }
    \qquad
    \mu:=\min_{1\le i\le n}
    \frac{a_i b_i}{(b_i+R)^3}>0,
\end{equation}
the pseudo-gradient of the game $\Phi$ satisfies
\begin{equation}\label{eq_strong_antimonotonicity}
    (x-y)\cdot\bigl(\Phi(x)-\Phi(y)\bigr)
    \le -\mu\lVert x-y\rVert^2
    \qquad \text{for all }x,y\in X.
\end{equation}
Here and below, the inner product and norm are Euclidean on $\R^{mn}$. In particular,
\[
    (x-y)\cdot\bigl(\Phi(x)-\Phi(y)\bigr)
    =
    \sum_{j=1}^m\sum_{i=1}^n
    (x_i^j-y_i^j)
    \bigl(c_i^j(x)-c_i^j(y)\bigr),
\]
and
\(
    \lVert x-y\rVert^2
    =
    \sum_{j=1}^m\sum_{i=1}^n
    (x_i^j-y_i^j)^2.
\)
Thus, the inequality in \eqref{eq_strong_antimonotonicity} is strict whenever $x\ne y$; hence, the pseudo-gradient is strictly anti-monotone, and the game is diagonally strictly concave in the sense of Rosen, with identity weights.
\end{Lem}

The following proposition is the formal statement of the existence and uniqueness assertion of part~\ref{thm_Equimarginal} of Theorem~\ref{thm_main_structure}. It also provides the main practical form of the equilibrium conditions: at equilibrium, each player equalizes marginal utility over all projects she actually uses, while every unused project has marginal utility no larger than this common value.

\begin{Prop}[Existence, uniqueness, and equimarginal characterization of equilibrium]\label{thm_existence_uniqueness}
The game $M(n,m,a,b,r)$ admits a unique Nash equilibrium. Moreover, if $x\in X$ is an admissible profile, then $x$ is the Nash equilibrium if and only if for every player $j$ there exists a number $c^j>0$ such that
\begin{equation}\label{equimarginal_characterization}
    c_i^j(x)=c^j \qquad \text{for all } \ \ i\in A^j(x),
    \qquad
    \text{and}
    \qquad 
    c_i^j(x)\le c^j \qquad \text{for all } \ \ i\notin A^j(x).
\end{equation}
Equivalently, for every player $j$,
\(
    c_i^j(x)\le c^j,
\) and 
\(
    x_i^j\bigl(c^j-c_i^j(x)\bigr)=0,
    \, i=1,\ldots,n.
\)
\end{Prop}

\begin{proof}
We start by proving the existence and uniqueness of Nash equilibrium. By Lemma~\ref{lem_concavity_antimonotonicity}, each payoff is continuous and concave in the player's own strategy, and the pseudo-gradient satisfies Rosen's diagonal strict concavity condition. Therefore, existence and uniqueness follow directly from Rosen~\cite[Theorems~1 and~2, p.~524]{Rosen}. For completeness, we also spell out the short argument in Appendix \ref{appendix_existence_uniqueness}.

It remains to prove the characterization \eqref{equimarginal_characterization}. Fix a player $j$ and fix the strategies $x^{-j}$ of all other players. Consider the problem
\[
    \max_{y^j\in X^j} F^j(y^j,x^{-j}).
\]
Since $F^j(\cdot,x^{-j})$ is concave, a vector $x^j\in X^j$ solves this problem if and only if
\(
    \nabla_{x^j}F^j(x)\cdot (y^j-x^j)\le0
\)
for all $y^j\in X^j$.

We now show that this first-order condition is equivalent to the stated marginal-utility conditions.

Assume first that $x^j$ is optimal. Since $r^j>0$, the active set $A^j(x)$ is nonempty. If two active projects $i,k\in A^j(x)$ had different marginal utilities, say $c_i^j(x)>c_k^j(x)$, then moving a sufficiently small amount of resource from project $k$ to project $i$ would increase the first-order variation of $F^j$, contradicting optimality. Hence, all active projects have a common marginal utility; denote it by $c^j$. If some inactive project $i\notin A^j(x)$ had $c_i^j(x)>c^j$, then moving a sufficiently small amount of resource from any active project to project $i$ would again increase the first-order variation, a contradiction. Therefore, $c_i^j(x)\le c^j$ for every inactive $i$.

Conversely, suppose that these marginal-utility conditions hold. Let $y^j\in X^j$ be arbitrary. Write $A^j=A^j(x)$. Since $x_i^j=0$ for $i\notin A^j$, we have
\[
\begin{aligned}
    \nabla_{x^j}F^j(x)\cdot (y^j-x^j)
    &=\sum_{i\in A^j} c^j(y_i^j-x_i^j)
      +\sum_{i\notin A^j} c_i^j(x)y_i^j  \\
    &\le c^j\sum_{i\in A^j}(y_i^j-x_i^j)
      +c^j\sum_{i\notin A^j}y_i^j  =c^j\left(\sum_{i=1}^n y_i^j-\sum_{i=1}^n x_i^j\right)=0.
\end{aligned}
\]
By concavity, this implies
\(
    F^j(y^j,x^{-j})\le F^j(x^j,x^{-j}).
\)
Thus, $x^j$ is a best response to $x^{-j}$. Since the argument applies to every player $j$, the profile $x$ is a Nash equilibrium. As has already been proved, this equilibrium is unique.
\end{proof}

\begin{Cor}[Variational-inequality and projection characterization]\label{cor_vi_projection}
An admissible profile $x\in X$ is the Nash equilibrium if and only if
\begin{equation}\label{eq_vi_characterization}
    \Phi(x)\cdot(y-x)\le0
    \qquad \text{for every }y\in X.
\end{equation}
Let $\Pi_X:\mathbb R^{mn}\to X$ denote the Euclidean projection onto $X$. Then, for every $\nu>0$, condition \eqref{eq_vi_characterization} is equivalent to
\begin{equation}\label{eq_projection_fixed_point_characterization}
    x=\Pi_X\bigl(x+\nu\Phi(x)\bigr).
\end{equation}
Thus, the Nash equilibrium is the unique fixed point of the projected pseudo-gradient map for every positive step size.
\end{Cor}

\begin{proof}
By concavity in each player's own strategy, $x$ is a Nash equilibrium if and only if, for every player $j$,
\(
    \nabla_{x^j}F^j(x)\cdot(y^j-x^j)\le0
\)
for every $y^j\in X^j$. Summing these inequalities over the players gives \eqref{eq_vi_characterization}. Conversely, by choosing $y\in X$ that differs from $x$ in only one player's coordinates, \eqref{eq_vi_characterization} recovers the individual first-order condition for that player, and hence the Nash property.

For a nonempty closed convex set $X$, the Euclidean projection is characterized by
\[
    p=\Pi_X(z)
    \quad\Longleftrightarrow\quad
    (z-p)\cdot(y-p)\le0
    \qquad \text{for every }y\in X.
\]
Taking $z=x+\nu\Phi(x)$ and $p=x$ proves the equivalence of \eqref{eq_vi_characterization} and \eqref{eq_projection_fixed_point_characterization}.
\end{proof}

\begin{Cor}[Active-project formula and boundary condition]\label{cor_active_formula_boundary}
Let $x^\ast$ be the Nash equilibrium. For every player $j$ and every project $i$ active for this player,
\begin{equation}\label{equilibrium_x_formula}
    x_i^{j,\ast}
    =L_i(x^\ast)-\frac{c^jL_i(x^\ast)^2}{a_i}
    =L_i(x^\ast)\left(1-\frac{c^jL_i(x^\ast)}{a_i}\right).
\end{equation}
Consequently,
\(
x_i^{j,\ast}>0
\)
implies 
\(
c^j<{a_i}/{L_i(x^\ast)}.
\)
For every project $i$ inactive for player $j$,
\[
    c^j\ge c_i^j(x^\ast)=\frac{a_i}{L_i(x^\ast)}.
\]
\end{Cor}

\begin{proof}
For an active project, Proposition~\ref{thm_existence_uniqueness} and \eqref{marginal_utility_formula} give $c^j=a_i(L_i(x^\ast)-x_i^{j,\ast})/L_i(x^\ast)^2$. Solving for $x_i^{j,\ast}$ yields \eqref{equilibrium_x_formula}, whose positivity implies $c^j<a_i/L_i(x^\ast)$. If the project is inactive, then $c_i^j(x^\ast)=a_i/L_i(x^\ast)$, and the boundary inequality follows from Proposition~\ref{thm_existence_uniqueness}.
\end{proof}

This characterization is deliberately stated without using cutoff or zone notation. In the next section, we solve the fundamental one-zone case, where all players are known a priori to be active on the same set of projects. The general cutoff and multi-zone structure is developed afterward.

\section{Local Structure: The Fully Active Regime}\label{section_local_structure}

We now turn from the existence and characterization of equilibrium to the structure of the equilibrium itself. The general equilibrium may have several active zones: some players may invest in some projects but not in others. The fundamental building block, however, is the case of a single active zone, where all players allocate positive investments in the same set of projects; this case is also the local model from which the general multi-zone structure will be assembled. Moreover, the formulas obtained here will later be applied zone by zone and form the algebraic basis of Block Pandora; they also provide exact local benchmarks for the projected algorithm.

Nothing in the calculations below requires the common active set to contain every project. If all players are active on the same fixed subset $I\subseteq\{1,\ldots,n\}$, with $|I|=K$, then the analysis applies verbatim to the restricted $K$-project game; the omitted projects are subsequently checked through the boundary inequalities of Corollary~\ref{cor_active_formula_boundary}. To avoid carrying the additional index set $I$, we take $I=\{1,\ldots,n\}$ throughout and refer to this regime as fully active.

Accordingly, throughout this section, we suppress the superscript $\ast$ from equilibrium quantities and assume that the unique Nash equilibrium \(x = x^\ast\) of \(M(n,m,a,b,r)\) is \textit{fully active}, that is,
\begin{equation}\label{eq_one_zone_standing_assumption}
    x_i^j>0
    \qquad
    \text{for all }\quad i=1,\ldots,n,\quad j=1,\ldots,m.
\end{equation}
By Proposition~\ref{thm_existence_uniqueness}, each player $j$ then has a common equilibrium marginal utility rate $c^j>0$ on all projects. We define the aggregate marginal utility rate, total resource, and aggregate baseline, respectively, by
\[
    C:=\sum_{j=1}^m c^j,
    \qquad
    R:=\sum_{j=1}^m r^j,
    \qquad
    B:=\sum_{i=1}^n b_i.
\]
Our goal is to determine \(C\) from a scalar nonlinear equation involving only \(R\), \(B\), and the project parameters; once \(C\) is known, all individual rates \(c^j\), loads \(L_i\), allocations \(x_i^j\), and payoffs are explicit.

We start by recalling \eqref{equilibrium_x_formula} of Corollary~\ref{cor_active_formula_boundary}, which, in the fully active regime, gives
\(
    x_i^j=L_i-{c^jL_i^2}/{a_i}.
\)
Summing over the players and using $\sum_j x_i^j=L_i-b_i$ yields
\begin{equation}\label{eq_one_zone_load_quadratic}
    \frac{C}{a_i}L_i^2-(m-1)L_i-b_i=0.
\end{equation}
Hence, $L_i$ is the positive root
\begin{equation}\label{Li_solution}
    L_i=G_i(C)
    :=\frac{a_i}{2C}
    \left((m-1)+\sqrt{(m-1)^2+\frac{4b_iC}{a_i}}\right),
    \qquad i=1,\ldots,n.
\end{equation}
Since $\sum_iL_i=R+B$, we obtain the \textbf{basic nonlinear equation}
\begin{equation}\label{eq_basic_nonlinear_equation}
    R+B=G(C):=\sum_{i=1}^nG_i(C)=\sum_{i=1}^n\frac{a_i}{2C}\left((m-1)+\sqrt{(m-1)^2+\frac{4b_iC}{a_i}}\right).
\end{equation}
Its solvability properties are recorded next and proved in Appendix~\ref{appendix_lem_basic_nonlinear_equation}.

\begin{Lem}[Solvability of the basic nonlinear equation]\label{lem_basic_nonlinear_equation}
For fixed positive parameter vectors $a,b$ and every $R>0$, equation \eqref{eq_basic_nonlinear_equation} has a unique solution $C>0$. Moreover, for fixed $a,b,n,m$, this solution is a strictly decreasing function of $R$.
\end{Lem}

Before we proceed to the main result of this section, we highlight that equation~\eqref{eq_basic_nonlinear_equation} relates the cumulative resource $R$ and the cumulative marginal utility rate $C$. It does not involve the individual resources $r^1,\ldots,r^m$ or the individual marginal rates $c^1,\ldots,c^m$. Consequently, as long as the equilibrium remains fully active, redistributing a fixed total resource $R$ among the players leaves $C$ unchanged; the individual marginal rates and allocations are recovered only afterward. We find this invariance striking from an intuitive and economic perspective. We are not aware of a closely analogous result or a compelling explanation for it, and regard understanding this phenomenon as an interesting question in its own right.

The following proposition is the formal statement of
part~\ref{thm_one_zone} of Theorem~\ref{thm_main_structure}.

\begin{Prop}[One-zone equilibrium formulas]\label{thm_one_zone_formulas}
Assume that the Nash equilibrium is fully active, in the sense of \eqref{eq_one_zone_standing_assumption}. Let \(c^1, \dots, c^m\) be players' equilibrium marginal utility rates, which exist by Proposition \ref{thm_existence_uniqueness}, and let
\(
    C=\sum_{j=1}^m c^j.
\)
Then \(C\) is the unique positive solution of the basic nonlinear equation \eqref{eq_basic_nonlinear_equation} with the given problem data $a, b, r$, and the equilibrium loads are given by
\begin{equation}\label{eq_one_zone_loads}
    L_i=G_i(C),
    \qquad i=1,\ldots,n,
\end{equation}
in the notation of \eqref{Li_solution}. 

Define
\begin{equation}\label{eq_one_zone_beta_definition}
    D:=(m-1)R+mB,
    \qquad
    \beta^j:=\frac{R+B-r^j}{D},
    \qquad j=1,\ldots,m.
\end{equation}
Then $\beta^j>0$, $\sum_j\beta^j=1$, and
\begin{equation}\label{eq_one_zone_cj_formula}
    c^j=\beta^jC,
    \qquad j=1,\ldots,m.
\end{equation}
The equilibrium allocations are
\begin{equation}\label{eq_one_zone_allocation_formula}
    x_i^j
    =L_i-\frac{c^jL_i^2}{a_i}
    =L_i-\beta^j\bigl((m-1)L_i+b_i\bigr),
    \qquad i=1,\ldots,n,\quad j=1,\ldots,m,
\end{equation}
and the project and total payoffs are
\begin{equation}\label{eq_one_zone_payoff_formula}
    f_i^j=a_i-c^jL_i,
    \quad
    F^j=\sum_{i=1}^na_i-c^j(R+B),
    \quad \text{ thus also } \quad
    \sum_{j=1}^mF^j
    =m\sum_{i=1}^na_i-C(R+B).
\end{equation}

Conversely, for arbitrary positive parameters $a,b,r$, let $C$ solve \eqref{eq_basic_nonlinear_equation} and define $L_i,D,\beta^j,c^j,x_i^j$ by \eqref{eq_one_zone_loads}--\eqref{eq_one_zone_allocation_formula}. If every $x_i^j$ is positive, then the resulting profile is the unique Nash equilibrium.
\end{Prop}

\begin{proof}
The load formula and the basic nonlinear equation were derived above, and Lemma~\ref{lem_basic_nonlinear_equation} gives the uniqueness of $C$. From \eqref{eq_one_zone_load_quadratic} and $\sum_iL_i=R+B$, we obtain
\[
    C\sum_{i=1}^n\frac{L_i^2}{a_i}
    =\sum_{i=1}^n\bigl((m-1)L_i+b_i\bigr)
    =(m-1)(R+B)+B=D.
\]
Summing the active-project formula over $i$ therefore gives
\[
    r^j
    =R+B-c^j\sum_{i=1}^n\frac{L_i^2}{a_i}
    =R+B-c^j\frac{D}{C},
\]
which is equivalent to \eqref{eq_one_zone_cj_formula}. Moreover,
\(
    R+B-r^j=B+\sum_{\ell\ne j}r^\ell>0
\)
and 
\(
    \sum_{j=1}^m(R+B-r^j)=D,
\)
so the weights in \eqref{eq_one_zone_beta_definition} are positive and sum to one. Formula \eqref{eq_one_zone_allocation_formula} follows from
\[
    \frac{CL_i^2}{a_i}=(m-1)L_i+b_i,
    \quad 
    \text{ and }
    \quad 
    f_i^j=\frac{a_ix_i^j}{L_i}=a_i-c^jL_i.
\]
Summing first over projects and then over players proves \eqref{eq_one_zone_payoff_formula}.

For the converse, suppose that the quantities in the statement are defined and that all $x_i^j>0$. Using $\sum_iL_i=R+B$, $\sum_j\beta^j=1$, and the definition of $D$, we obtain
\[
\begin{aligned}
    \sum_{i=1}^n x_i^j
    =R+B-\beta^jD=r^j
    \quad \text{ and } \quad
    b_i+\sum_{j=1}^m x_i^j
    =b_i+mL_i-\bigl((m-1)L_i+b_i\bigr)
      \sum_{j=1}^m\beta^j
      =L_i.
\end{aligned}
\]
Thus, the profile is admissible and the prescribed $L_i$ are its actual loads. Finally,
\[
    \frac{a_i(L_i-x_i^j)}{L_i^2}
    =\beta^jC=c^j
    \qquad \text{for all }i,j.
\]
The equimarginal characterization in Proposition~\ref{thm_existence_uniqueness} now shows that the profile is a Nash equilibrium, and uniqueness follows from the same proposition.
\end{proof}

\begin{remark}[Checking the fully active candidate]
\label{rem_full_activity_check}
The positivity condition in the converse part of Proposition~\ref{thm_one_zone_formulas} can be checked using only the least favorable player--project pair. Indeed, put
\(
    p_i:=a_i / L_i.
\)
By \eqref{eq_one_zone_allocation_formula}, the one-zone candidate satisfies $x_i^j>0$ if and only if $c^j<p_i$. Moreover, \eqref{eq_one_zone_beta_definition}--\eqref{eq_one_zone_cj_formula} give
\(
    c^1\le\cdots\le c^m.
\)
Rewriting \eqref{eq_one_zone_load_quadratic} as
\(
    q_i
    =
    {p_i^2}/(C-(m-1)p_i)
\)
shows that $p_i$ is increasing in $q_i$; hence $p_1\ge\cdots\ge p_n$. Consequently, all entries of the one-zone candidate are positive if and only if
\(
    c^m< a_n / L_n.
\)
Thus, this single inequality determines whether the candidate of Proposition~\ref{thm_one_zone_formulas} is the fully active Nash equilibrium.
\end{remark}

\begin{Cor}[Local monotonicity between players]\label{cor_one_zone_local_monotonicity}
Assume that the one-zone equilibrium is fully active and that
\(
    r^1\ge r^2\ge\cdots\ge r^m.
\)
Then
\(
    c^1\le c^2\le\cdots\le c^m,
\)
and, for every project $i=1,\ldots,n$, we have
\(
x_i^1\ge x_i^2\ge\cdots\ge x_i^m
\)
and
\(
f_i^1\ge f_i^2\ge\cdots\ge f_i^m.
\)
Moreover,
\(
    F^1\ge F^2\ge\cdots\ge F^m.
\)
\end{Cor}

\begin{proof}
By \eqref{eq_one_zone_beta_definition}--\eqref{eq_one_zone_cj_formula}, $c^j$ decreases as $r^j$ increases. The formulas \eqref{eq_one_zone_allocation_formula} and \eqref{eq_one_zone_payoff_formula} then show that $x_i^j$, $f_i^j$, and $F^j$ all decrease as $c^j$ increases.
\end{proof}

\begin{remark}[The case $m=1$]\label{rem_one_zone_m_one}
The formulas above include the one-player marginal utility model.  If $m=1$, then $C=c^1$, $\beta^1=1$, and
\[
    G_i(C)=\sqrt{\frac{a_i b_i}{C}}.
\]
The basic nonlinear equation becomes
\[
    R+B=\frac{1}{\sqrt C}\sum_{i=1}^n\sqrt{a_i b_i},
    \qquad
    \sqrt C=\frac{\sum_{i=1}^n\sqrt{a_i b_i}}{R+B}.
\]
For $b_i\equiv1$ and $R=1$, this gives
$$
\sqrt{C}=\frac{\sum_{i=1}^{n}\sqrt{a_i}}{n+1},
$$
which is equivalent, in the fully active case, to the closed form allocation formula obtained by Malinovsky and Sonin ~\cite{SoninMalinovsky}.
\end{remark}

The next section explains how the same one-zone formulas are used when the full equilibrium is split into several activity zones.

\section{Global Qualitative Structure}\label{section_global_structure}

In the previous section we analyzed the fundamental case in which the Nash equilibrium is fully active: every player invests in every project. In the general case this need not be true. Some players may stop investing before others, and some projects may be used only by the players with larger resources.

We do not have a closed analytic formula that determines the global active structure directly from the parameters of the model. Accordingly, the treatment of the general case has complementary qualitative and computational parts. The present section shows that the active sets have a rigid ordered structure. Section~\ref{section_algorithms} then develops one unconditional algorithm and one structure-exploiting reconstruction algorithm. The projected marginal-utility method converges globally without requiring the active structure to be known in advance, while Block Pandora reconstructs and checks the equilibrium conditional on a proposed nested cutoff pattern. The latter is particularly natural when such a pattern is known or strongly suggested by the parameters, or when only a small number of player layers is expected.

We start by fixing the terminology. Suppose that a profile \(x\) has the property that, for every player \(j\), there exists an integer \(k^j\) such that
\[
    x_i^j>0 \quad \text{for } i\le k^j,
    \qquad
    x_i^j=0 \quad \text{for } i>k^j.
\]
We call \(k^j\) the \textit{cutoff} of player \(j\). If the cutoffs are nested,
\[
    n\ge k^1\ge k^2\ge \cdots \ge k^m\ge 1,
\]
then their distinct values define a partition of the active projects into consecutive \textit{zones}. More precisely, let
\[
    0=z_0<z_1<\cdots<z_N
\]
be the distinct positive values among the cutoffs \(k^1,\ldots,k^m\), written in increasing order. Define
\[
    Z_s:=\{z_{s-1}+1,\ldots,z_s\},
    \qquad s=1,\ldots,N.
\]
The set of players active in zone \(Z_s\) is denoted by
\[
    J_s:=\{j:k^j\ge z_s\},
    \quad \text{ thus } \qquad 
    J_1\supseteq J_2\supseteq\cdots\supseteq J_N.
\]
We also write
\[
    m_s:=|J_s|, \qquad B_s:=\sum_{i\in Z_s}b_i.
\]
If \(j\in J_s\), we denote the total amount invested by player \(j\) in zone \(Z_s\) and the total resource in the zone, respectively, by
\[
    r_s^j:=\sum_{i\in Z_s}x_i^j, \quad R_s:=\sum_{j\in J_s}r_s^j.
\]
It will also be useful later, in the algorithmic section, to use the dual language of layers and blocks. The players with the same cutoff form a \textit{layer}:
\[
    Y_s:=\{j:k^j=z_s\}, \qquad s=1,\ldots,N.
\]
A \textit{block} is the intersection of a player layer and a project zone. The block \(Y_t\times Z_s\) is active exactly when \(s\le t\).

The main result of this section is that the unique Nash equilibrium indeed has the ordered cutoff structure described above. Accordingly, the following proposition is the formal statement of part~\ref{thm_Ordered} of Theorem~\ref{thm_main_structure}.

\begin{Prop}[Cutoff and nested zone structure]\label{thm_cutoff_zone_structure}
Let \(x^\ast\) be the unique Nash equilibrium of the model \(M(n,m,a,b,r)\). Then there exist cutoffs
\(
    n\ge k^1\ge k^2\ge\cdots\ge k^m\ge 1
\)
such that, for every player \(j\),
\[
    x_i^{j,\ast}>0 \quad \text{for } i\le k^j,
    \qquad
    x_i^{j,\ast}=0 \quad \text{for } i>k^j.
\]
Equivalently, after the projects are ordered by their qualities \(q_1\ge q_2\ge\cdots\ge q_n\), each player invests only in an initial segment of the project list, and players with larger resources have weakly larger active segments.
\end{Prop}

\begin{proof}
Let \(c^j\) be the common equilibrium marginal utility rate of player \(j\), whose
existence was proved in Proposition~\ref{thm_existence_uniqueness}. Throughout
this proof we write
\begin{equation}\label{p_i_def}
    L_i:=L_i(x^\ast),
    \qquad
    p_i:=\frac{a_i}{L_i}.
\end{equation}
We call \(p_i\) the effective marginal level of project \(i\) at equilibrium.

By Corollary~\ref{cor_active_formula_boundary}, for every player \(j\) and project \(i\),
\[
    x_i^{j,\ast}>0
    \quad \Longleftrightarrow \quad
    c^j<p_i, 
    \qquad \text{ and } \qquad
    x_i^{j,\ast}=0
    \quad \Longleftrightarrow \quad
    c^j\ge p_i.
\]
Therefore, to obtain the cutoff structure claimed in the proposition, it suffices to show that the effective levels \(p_i\) inherit the ordering of the qualities \(q_i=a_i/b_i\). Indeed, assume that this has been established, and fix a player \(j\). Since player \(j\) is active in project \(i\) exactly when \(c^j<p_i\), and the sequence \((p_i)\) is non-increasing, the active set of player \(j\) must be an initial segment of projects. Thus, there exists \(k^j\) such that
\(
x_i^{j,\ast}>0
\)
for
\(
i\le k^j
\)
and
\(
x_i^{j,\ast}=0
\)
for
\(i>k^j.\)

It remains to prove that these cutoffs are nested according to the players' resources. Recall from Corollary~\ref{cor_active_formula_boundary} the positive-part formula
\begin{equation}\label{positive_part_formula_global}
    x_i^{j,\ast}
    =
    L_i\left(1-\frac{c^j}{p_i}\right)_+,
    \qquad
    p_i=\frac{a_i}{L_i},
\end{equation}
with the notation \(t_+=\max\{t,0\}\).
Using this identity, define
\[
    \Psi(c):=\sum_{i=1}^n L_i\left(1-\frac{c}{p_i}\right)_+,
    \quad \text{ so that } \quad 
    r^j=\sum_{i=1}^n x_i^{j,\ast}=\Psi(c^j).
\]
The function \(\Psi\) is continuous and strictly decreasing on the relevant range: between its finitely many breakpoints,
\(
    \Psi'(c)=-\sum_{i:c<p_i}L_i/p_i<0.
\)
As $\Psi(c^j)=r^j>0$, at least one summand is active at every equilibrium rate. Therefore,
\(
r^j\ge r^\ell
\)
implies 
\(
c^j\le c^\ell.
\)
If player \(\ell\) is active in project \(i\), then \(c^\ell<p_i\). Since \(c^j\le c^\ell\), we also have \(c^j<p_i\), and hence player \(j\) is active in project \(i\). Thus, the active set of player \(\ell\) is contained in the active set of player \(j\). Consequently,
\(
r^j\ge r^\ell
\)
implies
\(
k^j\ge k^\ell.
\)
Since the players were ordered so that \(r^1\ge r^2\ge\cdots\ge r^m\), we get
\(
    k^1\ge k^2\ge\cdots\ge k^m,
\)
which is the desired statement. Thus, it remains to show that if $q_1 \ge \dots \ge q_n$, then $p_1 \ge \dots \ge p_n$ as well.

Fix the equilibrium rates \(c^1,\ldots,c^m\) and set \(F(p):=\sum_{j=1}^m(1-c^j/p)_+\). For \(q>0\), let \(p(q)\) be the unique solution of
\begin{equation}\label{effective_level_equation}
    F(p)=1-\frac{p}{q}.
\end{equation}
Existence and uniqueness follow because \(F\) is continuous and non-decreasing, whereas the right-hand side is strictly decreasing, with \(F(p)<1-p/q\) for \(p>0\) sufficiently small and \(F(q)\ge0=1-q/q\). For the actual project \(i\), summing \eqref{positive_part_formula_global} gives \(F(p_i)=(L_i-b_i)/L_i=1-p_i/q_i\), so \(p_i=p(q_i)\). Now let \(q'>q\) and put \(p=p(q)\). Then \(F(p)=1-p/q<1-p/q'\). Since \(F\) is non-decreasing and \(p\mapsto1-p/q'\) is strictly decreasing, their unique intersection satisfies \(p(q')>p(q)\). Therefore, \(q_1\ge\cdots\ge q_n\) implies \(p_1\ge\cdots\ge p_n\), which concludes the proof.
\end{proof}

Proposition~\ref{thm_cutoff_zone_structure} shows that the general equilibrium is not an arbitrary pattern of active and inactive entries. It is organized by consecutive zones and nested active player sets. Once the zone structure and the amounts \(r_s^j\) invested by each active player in each zone are known, the local formulas from Section~\ref{section_local_structure} apply inside every zone.

\begin{Cor}[Local form inside each zone]\label{cor_local_form_inside_zones}
Let \(x^\ast\) be the Nash equilibrium, and let
\(
    Z_1,\ldots,Z_N,
    \,
    J_1\supseteq J_2\supseteq\cdots\supseteq J_N
\)
be the zones and active player sets defined by Proposition~\ref{thm_cutoff_zone_structure}. Fix a zone \(Z_s\). If the zone resources
\(
    r_s^j=\sum_{i\in Z_s}x_i^{j,\ast},
    \, j\in J_s,
\)
are known, then the restriction of \(x^\ast\) to the projects in \(Z_s\) is the fully active one-zone equilibrium for the subgame with projects \(Z_s\), active players \(J_s\), resources \((r_s^j)_{j\in J_s}\), and project parameters \((a_i,b_i)_{i\in Z_s}\).

Consequently, all loads, marginal utility rates, allocations, and payoffs inside \(Z_s\) are obtained by applying Proposition~\ref{thm_one_zone_formulas} to that subgame.
\end{Cor}

\begin{proof}
Fix a zone \(Z_s\). By definition of the zones, every player in \(J_s\) is active in every project of \(Z_s\), and every player outside \(J_s\) invests zero in every project of \(Z_s\).

Now fix the total amounts \(r_s^j\), \(j\in J_s\), invested in the zone. If the restriction of \(x^\ast\) to \(Z_s\) were not a Nash equilibrium of the corresponding fully active subgame, then some player \(j\in J_s\) could improve her payoff by redistributing only the amount \(r_s^j\) among the projects of \(Z_s\), leaving all investments outside \(Z_s\) unchanged. This would be an admissible deviation in the original game and would contradict the Nash property of \(x^\ast\).

Thus, the restriction to \(Z_s\) is the fully active equilibrium of the local subgame, and Proposition~\ref{thm_one_zone_formulas} applies.
\end{proof}

The corollary is local rather than closed form, since the zones and the resources $r_s^j$ are unknown and must satisfy $c_s^j=c_t^j$ whenever player $j$ is active in both zones. Block Pandora imposes these smooth-fit conditions explicitly, whereas the projected marginal-utility method converges directly to the global equilibrium, whose active zones automatically satisfy these equalities.

We conclude this section by recording the monotonicity properties of the equilibrium. The following proposition is the formal statement of
part~\ref{thm_monotonicity} of Theorem~\ref{thm_main_structure}.

\begin{Prop}[Monotonicity inside equilibrium]\label{thm_global_monotonicity}
Let \(x^\ast\) be the Nash equilibrium.

\begin{enumerate}[label=\textup{(\roman*)}]
    \item If \(r^j\ge r^\ell\), then
    \(
        c^j\le c^\ell,
        \,
        k^j\ge k^\ell,
        \,
        \text{and }
        x_i^{j,\ast}\ge x_i^{\ell,\ast}
    \)
    for every $i=1,\ldots,n$.
    
    \item Assume, in addition, that the congestion parameters are ordered as
    \(
        b_1\ge b_2\ge\cdots\ge b_n.
    \)
    Then, for every player \(j\),
    \(
        x_1^{j,\ast}\ge x_2^{j,\ast}\ge\cdots\ge x_n^{j,\ast}.
    \)
    \label{thm_global_monotonicity_ii}
\end{enumerate}
\end{Prop}

\begin{proof}
The inequalities $c^j\le c^\ell$ and $k^j\ge k^\ell$ in part~\textup{(i)} were proved in Proposition~\ref{thm_cutoff_zone_structure}. The positive-part formula
\begin{equation}\label{positive-part_formula}
    x_i^{j,\ast}
    =
    L_i\left(1-\frac{c^j}{p_i}\right)_+, \quad i = 1, \dots, n,
\end{equation}
with $p_i$ defined in \eqref{p_i_def},
is non-increasing in $c^j$, and therefore also gives $x_i^{j,\ast}\ge x_i^{\ell,\ast}$.

For part~\textup{(ii)}, Proposition~\ref{thm_cutoff_zone_structure} gives $p_i\ge p_{i+1}$. On an interval where exactly $s$ players are active, \eqref{effective_level_equation} yields
\[
    \frac{q}{p}=\frac{p}{S_s-(s-1)p},
\]
where $S_s$ is the sum of the $s$ active marginal rates. This expression is increasing in $p$, and the pieces fit continuously at the activity thresholds. Hence,
\(
    q_i/p_i\ge q_{i+1}/p_{i+1}.
\)
Since $L_i=b_iq_i/p_i$ and $b_i\ge b_{i+1}$, it follows that $L_i\ge L_{i+1}$. Together with $p_i\ge p_{i+1}$ and the positive-part formula \eqref{positive-part_formula}, this gives $x_i^{j,\ast}\ge x_{i+1}^{j,\ast}$, which concludes the proof.
\end{proof}

\begin{remark}
The additional monotonicity assumption on \(b\) in part (ii) is essential. If the projects are ordered by \(q_i=a_i/b_i\) but the values \(b_i\) are not ordered, then projectwise monotonicity of the allocations can fail already in the one-player model.
\end{remark}

\section{Algorithms}\label{section_algorithms}

The most immediate iterative procedure for finding the Nash equilibrium is the best response algorithm, in which players are updated sequentially using the most recently available strategies of their opponents. In our numerical experiments this procedure often converges rapidly, but it is not globally convergent: the final example in Section~\ref{section_examples} provides a two-player instance in which almost every initial profile approaches a strict two-cycle. Thus, uniqueness of equilibrium and diagonal strict concavity do not by themselves imply convergence of best-response dynamics. This failure motivates the two algorithms developed below.

\subsection{The Projected Marginal-Utility Algorithm}\label{subsection_projected_marginal_utility}

We first develop a globally convergent algorithm that does not require the equilibrium active sets, cutoffs, or zones to be known in advance. At every iteration, each player moves in the direction of her current marginal utilities and then projects the tentative allocation back onto her resource simplex. The projection keeps every iterate admissible and automatically permits coordinates to enter or leave the active set.

The update resembles projected gradient ascent, but it is not, in general, gradient ascent for a common scalar objective. Indeed, for distinct players $j\ne\ell$ and a fixed project $i$, differentiation of \eqref{marginal_utility_formula} gives
\[
    \frac{\partial c_i^j}{\partial x_i^\ell}(x)
    =\frac{a_i(2x_i^j-L_i(x))}{L_i(x)^3},
    \qquad
    \frac{\partial c_i^\ell}{\partial x_i^j}(x)
    =\frac{a_i(2x_i^\ell-L_i(x))}{L_i(x)^3},
\]
and these cross-partials need not agree. Thus, the pseudo-gradient $\Phi$ of \eqref{def_Phi} is generally not the gradient of a potential function. The convergence result below instead relies on the quantitative anti-monotonicity established in Lemma~\ref{lem_concavity_antimonotonicity}, together with the global Lipschitz estimate in Lemma \ref{lem_pseudogradient_lipschitz}, proved in Appendix \ref{appendix_lem_pseudogradient_lipschitz}.

\begin{Lem}[Global Lipschitz continuity of the pseudo-gradient]\label{lem_pseudogradient_lipschitz}
The pseudo-gradient $\Phi$ of \eqref{def_Phi} is globally Lipschitz on $X$, defined in \eqref{def_X_j}--\eqref{def_X}. In particular, with
\begin{equation}\label{eq_pseudogradient_lipschitz_constant}
    K_\Phi:=(m+1)\max_{1\le i\le n}\frac{a_i}{b_i^2},
\end{equation}
we have
\begin{equation}\label{eq_pseudogradient_lipschitz}
    \lVert\Phi(x)-\Phi(y)\rVert
    \le K_\Phi\lVert x-y\rVert
    \qquad \text{for all }x,y\in X.
\end{equation}
\end{Lem}

Recall that $\Pi_X:\R^{mn}\to X$ denotes the Euclidean projection onto $X$. For $\nu>0$, define the projected marginal-utility map
\begin{equation}\label{eq_projected_map_definition}
    P_\nu(x):=\Pi_X\bigl(x+\nu\Phi(x)\bigr),
    \qquad x\in X.
\end{equation}
Because $X=X^1\times\cdots\times X^m$, the projection separates across players:
\[
    P_\nu(x)^j
    =\Pi_{X^j}\bigl(x^j+\nu\nabla_{x^j}F^j(x)\bigr),
    \qquad j=1,\ldots,m.
\]
Starting from an arbitrary profile $x^{(0)}\in X$, the algorithm is the simultaneous iteration
\begin{equation}\label{eq_projected_marginal_utility_iteration}
    x^{(k+1)}=P_\nu(x^{(k)})
    =\Pi_X\bigl(x^{(k)}+\nu\Phi(x^{(k)})\bigr),
    \qquad k=0,1,2,\ldots.
\end{equation}
The word \emph{simultaneous} is important: all players form their updates from the same current profile $x^{(k)}$ before any coordinate is replaced.

The playerwise projection in \eqref{eq_projected_marginal_utility_iteration} is explicit up to one scalar threshold. If
\[
    y_i^{j,(k)}:=x_i^{j,(k)}+\nu c_i^j(x^{(k)}),
    \quad \text{ then } \quad
    x_i^{j,(k+1)}
    =\max\{y_i^{j,(k)}-\tau_j^{(k)},0\},
\]
where $\tau_j^{(k)}$ is the unique number satisfying
\(
    \sum_{i=1}^n\max\{y_i^{j,(k)}-\tau_j^{(k)},0\}=r^j.
\) 
Thus, every iteration consists only of evaluating the current marginal utilities and projecting one vector onto a simplex for each player. The threshold $\tau_j^{(k)}$ can be found by sorting the $n$ coordinates of $y^{j,(k)}$, and the $m$ playerwise projections are independent and can be carried out in parallel.

The following proposition gives the globally convergent construction in part~\ref{thm_algo} of Theorem~\ref{thm_main_structure}.

\begin{Prop}[Global linear convergence]\label{thm_projected_algorithm_convergence}
Let $\mu$, $K_\Phi$ be as in \eqref{def_R_mu}, \eqref{eq_pseudogradient_lipschitz_constant}, respectively.
If
\begin{equation}\label{eq_projected_step_size_condition}
    0<\nu<\frac{2\mu}{K_\Phi^2},
\end{equation}
then $P_\nu$ is a contraction on $X$ with contraction factor
\begin{equation}\label{eq_projected_contraction_factor}
    q_\nu:=\sqrt{1-2\nu\mu+\nu^2K_\Phi^2}<1.
\end{equation}
Consequently, for every initial profile $x^{(0)}\in X$, the sequence generated by \eqref{eq_projected_marginal_utility_iteration} converges to the unique Nash equilibrium $x^\ast$, and
\begin{equation}\label{eq_projected_linear_rate}
    \lVert x^{(k)}-x^\ast\rVert
    \le q_\nu^k\lVert x^{(0)}-x^\ast\rVert,
    \qquad k=0,1,2,\ldots.
\end{equation}
\end{Prop}

\begin{proof}
Since $X$ is a nonempty closed convex subset of $\R^{mn}$, the Euclidean projection $\Pi_X:\R^{mn}\to X$ is nonexpansive; that is,
\(
    \lVert\Pi_X(u)-\Pi_X(v)\rVert\le\lVert u-v\rVert
    \,\text{ for all }u,v\in\R^{mn}.
\)
Hence, for any $x,y\in X$,
\begin{align*}
    \lVert P_\nu(x)-P_\nu(y)\rVert^2
    &\le
    \bigl\lVert (x-y)+\nu\bigl(\Phi(x)-\Phi(y)\bigr)\bigr\rVert^2\\
    &=\lVert x-y\rVert^2
      +2\nu(x-y)\cdot\bigl(\Phi(x)-\Phi(y)\bigr)
      +\nu^2\lVert\Phi(x)-\Phi(y)\rVert^2\\
    &\le
    \bigl(1-2\nu\mu+\nu^2K_\Phi^2\bigr)\lVert x-y\rVert^2,
\end{align*}
where the last line uses \eqref{eq_strong_antimonotonicity} and \eqref{eq_pseudogradient_lipschitz}. Under \eqref{eq_projected_step_size_condition}, the coefficient is strictly smaller than one, so $P_\nu:X\to X$ is a contraction with factor $q_\nu$.

Since $X$ is closed in the finite-dimensional Euclidean space $\R^{mn}$, it is complete. The Banach fixed-point theorem therefore gives a unique fixed point of $P_\nu$ and convergence of the iterates at the rate \eqref{eq_projected_linear_rate}. By Corollary~\ref{cor_vi_projection}, the fixed points of $P_\nu$ are exactly the Nash equilibria. Hence, the fixed point is $x^\ast$.
\end{proof}

Because $\mu>0$ and $K_\Phi<\infty$, the interval in \eqref{eq_projected_step_size_condition} is nonempty for every admissible parameter vector. A convenient fully explicit choice is
\(
    \nu={\mu}/{K_\Phi^2},
\)
which minimizes the upper bound in \eqref{eq_projected_contraction_factor} and gives
$q_\nu=\sqrt{1-\mu^2/K_\Phi^2}$.
The constants $\mu$ and $K_\Phi$ are deliberately global, so that the resulting step size can be conservative. Its advantage is that it is available directly from the model parameters and is valid regardless of which coordinates are active at equilibrium.

\subsection{The Block Pandora Algorithm}\label{section_block_pandora_algorithm}

The projected marginal-utility method treats the equilibrium problem as a monotone variational inequality and does not use the special form of the active sets.  We now describe a complementary, structure-exploiting algorithm. Conditional on a nested cutoff pattern, Block Pandora assembles the equilibrium from the one-zone formulas of Section~\ref{section_local_structure}.  The only nonlinear operations are one-dimensional smooth-fit equations, which can be solved by bisection.  The resulting construction is particularly useful when the equilibrium has few zones, since it works with blocks and layers rather than with all $mn$ allocation variables. We focus on the mathematical architecture of the procedure---its block decomposition, one-dimensional gluing equations, and certification step---rather than on implementation details, whose full treatment would require considerably more space.

Fix a candidate cutoff vector
\(
    \kappa=(k^1,\ldots,k^m),
    \,
    n\ge k^1\ge\cdots\ge k^m\ge1,
\)
and restrict player $j$ to the first $k^j$ projects:
\(
    X^j(\kappa)
    :=\left\{x^j\in X^j:x_i^j=0\ \text{for }i>k^j\right\}.
\)
The corresponding restricted game has a unique Nash equilibrium by the same argument as in Section~\ref{section_existence_equimarginal}.  We call it \emph{strict} if $x_i^j>0$ for every allowed coordinate $i\le k^j$.  In that case the zones $Z_s$, active player sets $J_s$, and layers $Y_s$ are exactly those introduced in Section~\ref{section_global_structure}.

The following identity is the algebraic basis of the algorithm.  Suppose that a zone $Z_s$ is solved as a fully active one-zone game with active players $J_s$ and zone resources $(r_s^j)_{j\in J_s}$.  Let $c_s^j$ be the resulting marginal rate of player $j$ in that zone, put
\[
    C_s:=\sum_{j\in J_s}c_s^j,
    \qquad
    \Lambda_s:=\sum_{i\in Z_s}L_i=R_s+B_s,
    \qquad
    W_s:=\sum_{i\in Z_s}\frac{L_i^2}{a_i}.
\]

\begin{Lem}[Affine block identity]\label{lem_block_affine_identity}
For every $j\in J_s$,
\begin{equation}\label{eq_block_affine_identity}
    r_s^j=\Lambda_s-W_sc_s^j.
\end{equation}
Moreover,
\begin{equation}\label{eq_block_W_formula}
    W_s=\frac{(m_s-1)R_s+m_sB_s}{C_s}.
\end{equation}
Consequently, if a layer $Y_t$ is smoothly fitted across the zones $Z_1,\ldots,Z_t$, then, with
\[
    \Lambda^{(t)}:=\sum_{s=1}^t\Lambda_s,
    \qquad
    W^{(t)}:=\sum_{s=1}^tW_s,
\]
we have
\begin{equation}\label{eq_layer_affine_identity}
    c^j=\frac{\Lambda^{(t)}-r^j}{W^{(t)}}
    \qquad (j\in Y_t).
\end{equation}
In particular, the sum of the marginal rates over any subset $E\subseteq Y_t$ depends on the resources of the players in $E$ only through their total:
\begin{equation}\label{eq_layer_subset_rate}
    \sum_{j\in E}c^j
    =\frac{|E|\Lambda^{(t)}-\sum_{j\in E}r^j}{W^{(t)}}.
\end{equation}
\end{Lem}

\begin{proof}
By Proposition~\ref{thm_one_zone_formulas}, inside $Z_s$,
\(
    x_i^j=L_i-c_s^j{L_i^2}/{a_i}.
\)
Summing over $i\in Z_s$ gives \eqref{eq_block_affine_identity}.  The local load equation
\(
    C_sL_i^2/a_i=(m_s-1)L_i+b_i
\)
gives \eqref{eq_block_W_formula} after summation over the zone.  Finally, summing \eqref{eq_block_affine_identity} over $s\le t$ and using the smooth-fit equalities $c_s^j=c^j$ proves \eqref{eq_layer_affine_identity} and \eqref{eq_layer_subset_rate}.
\end{proof}

Writing
\[
    \rho_t:=\sum_{j\in Y_t}r^j,
    \qquad
    \gamma_t:=\sum_{j\in Y_t}c^j,
\]
we have $C_s=\sum_{t=s}^N\gamma_t$, while \eqref{eq_layer_subset_rate} with $E=Y_t$ expresses $\gamma_t$ using only $|Y_t|$, $\rho_t$, and the aggregate zone quantities. Thus, for a fixed cutoff table, all aggregate loads and rates depend on the players in a layer only through the size and total resource of that layer. The individual resources enter only in the final reconstruction \eqref{eq_layer_affine_identity}.

\paragraph{Recursive gluing.}
For one zone, Block Pandora is simply the one-zone construction of Proposition~\ref{thm_one_zone_formulas}.  Suppose now that the candidate table has $N\ge2$ zones.  The rightmost active block is $Y_N\times Z_N$; only the players in $Y_N$ use $Z_N$.  Write
\[
    \ell:=|Y_N|,
    \qquad
    \rho:=\sum_{j\in Y_N}r^j,
\]
and let $y$ denote the total resource that these players assign to $Z_N$.

For a trial value of $y$, solve the one-zone problem on $Z_N$ with $\ell$ players and total resource $y$.  Its aggregate quantities will be denoted by
\[
    \Lambda_R(y),\qquad W_R(y),\qquad
    \Gamma_R(y):=\sum_{j\in Y_N}c_N^j.
\]
Here $\Gamma_R(y)$ is exactly the cumulative rate $C_N$ determined by the basic nonlinear equation.  Next solve, recursively, the table formed by the first $N-1$ zones, assigning the marked group $Y_N$ the total resource $\rho-y$ on those zones.  In this left table the layers $Y_{N-1}$ and $Y_N$ have the same cutoff and are merged for purposes of the aggregate calculation, although $Y_N$ remains marked.  Let
\[
    \Lambda_L(y):=\sum_{s=1}^{N-1}\Lambda_s(y),
    \qquad
    W_L(y):=\sum_{s=1}^{N-1}W_s(y).
\]
By \eqref{eq_layer_subset_rate}, the sum of the left marginal rates of the marked players is
\begin{equation}\label{eq_block_left_rate}
    \Gamma_L(y)
    =\frac{\ell\Lambda_L(y)-(\rho-y)}{W_L(y)}.
\end{equation}
The two sides fit precisely when
\begin{equation}\label{eq_block_gluing_equation}
    H(y):=\Gamma_L(y)-\Gamma_R(y)=0.
\end{equation}
We call a connected interval of trial values a \emph{strict branch} if all recursively defined left and right subproblems exist and remain strict throughout that interval. On every strict branch, $\Gamma_L$ and $\Gamma_R$ are continuous, $\Gamma_L$ is strictly increasing, and $\Gamma_R$ is strictly decreasing. Hence $H$ is continuous and strictly increasing and has at most one zero on that branch. Whenever $y_-<y_+$ belong to the same strict branch and
\(
    H(y_-)H(y_+)<0,
\)
bisection converges to the unique zero in $(y_-,y_+)$.

Let $y^\ast$ be this solution, and suppress its argument in the notation below.  The individual split of a player $j\in Y_N$ is then explicit:
\begin{equation}\label{eq_block_individual_rate}
    c^j
    =\frac{\Lambda_L+\Lambda_R-r^j}{W_L+W_R},
\end{equation}
\begin{equation}\label{eq_block_individual_split}
    r_N^j=\Lambda_R-W_Rc^j,
    \qquad
    r_L^j=r^j-r_N^j=\Lambda_L-W_Lc^j.
\end{equation}
Equation \eqref{eq_block_gluing_equation} guarantees that $\sum_{j\in Y_N}r_N^j=y^\ast$.  The right block is recovered from Proposition~\ref{thm_one_zone_formulas}, and the left allocations are recovered recursively.  If every prescribed active coordinate is positive, the result is the strict restricted equilibrium for the candidate cutoff vector; otherwise that candidate is discarded.

The following proposition gives the structure-exploiting construction in part~\ref{thm_algo} of Theorem~\ref{thm_main_structure}.

\begin{Prop}[Reconstruction for a fixed cutoff table]
\label{prop_block_pandora_correctness}
Fix a nested cutoff vector $\kappa$.

\begin{enumerate}[label=\textup{(\roman*)}]
    \item On every strict branch, the gluing equation
    \eqref{eq_block_gluing_equation} has at most one solution. If a sign-changing bracket on that branch is available, bisection converges to the unique solution.

    \item Every strict Nash equilibrium of the restricted game associated with $\kappa$ satisfies the recursive gluing equations and is recovered by the construction on the strict branch containing its block-resource vector. Conversely, whenever the construction produces a strict profile $x$, that profile is the unique Nash equilibrium of the restricted game.

    \item A strict restricted equilibrium $x$ is the Nash equilibrium of the original game if and only if it satisfies the omitted-project inequalities
    \begin{equation}\label{eq_block_boundary_certificate}
        c_i^j(x)=\frac{a_i}{L_i(x)}\le c^j,
        \qquad i>k^j.
    \end{equation}
\end{enumerate}
\end{Prop}

\begin{proof}
We argue by induction on the number of zones. As part of the induction hypothesis, we also show that, on every strict branch on which the recursive gluing roots exist, all aggregate quantities produced by the construction depend continuously on the external resource parameters.

The one-zone case follows from Lemma~\ref{lem_basic_nonlinear_equation} and the explicit formulas of Section~\ref{section_local_structure}. Suppose now that the assertion holds for all tables with at most $N-1$ zones, and consider a table with $N$ zones. The right-hand quantities $\Lambda_R(y)$, $W_R(y)$, and $\Gamma_R(y)$ are continuous in $y$ by the one-zone formulas. The left-hand quantities $\Lambda_L(y)$ and $W_L(y)$ are outputs of the recursive construction for an $(N-1)$-zone table with marked resource $\rho-y$, and are therefore continuous by the induction hypothesis. Hence $\Gamma_L$, $\Gamma_R$, and $H=\Gamma_L-\Gamma_R$ are continuous on every strict branch.

The monotonicity follows from the projectwise calculation used in the proof of Lemma~\ref{lem_concavity_antimonotonicity}. That calculation applies to two nonnegative profiles lying in a common bounded set, even when their resource vectors differ. If $x$ and $x'$ are strict restricted equilibria for the same active table, with resource vectors $r,r'$ and marginal rates $c^j,c^{j\prime}$, then
\(
    \sum_j(r^{j\prime}-r^j)(c^{j\prime}-c^j)
    =
    (x'-x)\cdot\bigl(\Phi(x')-\Phi(x)\bigr)
    <0.
\)
By the layer aggregation in \eqref{eq_layer_subset_rate}, a change in the total resource of the marked layer may be distributed equally among its players without changing its aggregate marginal rate. Hence the marked aggregate rate on the left strictly decreases with the resource assigned to the left table, and therefore $\Gamma_L(y)$ strictly increases with $y$. Lemma~\ref{lem_basic_nonlinear_equation} gives that $\Gamma_R(y)$ strictly decreases with $y$. Thus $H$ is strictly increasing, which proves part~\textup{(i)}.
Now let a strict restricted equilibrium be given. Its resource in the final block supplies a zero of $H$, because the marginal rates agree across the two sides. By part~\textup{(i)}, this zero is the only one on its strict branch, and formulas \eqref{eq_block_individual_rate}--\eqref{eq_block_individual_split}, followed by the induction hypothesis, recover the equilibrium. Conversely, a strict profile produced by the construction satisfies each player's budget constraint and equalizes that player's marginal rate over every allowed project. The equimarginal characterization therefore makes it a Nash equilibrium of the restricted game, which is unique by the argument of Section~\ref{section_existence_equimarginal}. This proves part~\textup{(ii)}.

Finally, the active-project equalities already hold for a strict restricted equilibrium. Adding \eqref{eq_block_boundary_certificate} gives exactly the full equimarginal conditions of Proposition~\ref{thm_existence_uniqueness}. This proves part~\textup{(iii)}.
\end{proof}

There are only
\(
    \binom{n+m-1}{m}
\)
weakly decreasing cutoff vectors. Exhaustive support enumeration, together with one-dimensional searches on the strict branches, therefore recovers the equilibrium in principle: for the true cutoff vector, every strict gluing root admits a rational sign-changing bracket, and the omitted-project inequalities identify the valid reconstruction. This establishes global recoverability by a finite support search, but does not assert a polynomial complexity bound. A complete implementation-level analysis is beyond the scope of the present paper, and the projected marginal-utility method remains preferable when the active structure is not expected to be sparse.

\section{Examples}\label{section_examples}

We now give four examples illustrating the fully active regime, the transition from the zero-baseline model to a multi-zone equilibrium, a genuine three-zone equilibrium, and the failure of cyclic best-response dynamics. Unless stated otherwise, numerical values are rounded to three decimal places.

\paragraph{A fully active equilibrium.}
Let
\[
    m=n=3,\qquad a=(12,7,3),\qquad b=(3,2,1),\qquad r=(2.2,2,1.8).
\]
Here $a$, $b$, and $q=a/b=(4,3.5,3)$ are all strictly decreasing. Moreover, the players' resources are deliberately chosen to be close to one another, so that a ``fully active'' equilibrium is expected a priori. Accordingly, Proposition~\ref{thm_one_zone_formulas} gives
\[
    C\approx4.582,\qquad
    (c^1,c^2,c^3)\approx(1.497,1.527,1.558),
    \quad \text{ and } \quad 
    x^\ast\approx
    \begin{pmatrix}
        1.258 & 0.681 & 0.261\\
        1.152 & 0.616 & 0.232\\
        1.046 & 0.552 & 0.203
    \end{pmatrix}.
\]
The smallest effective project level is
\(
    p_3={a_3}/{L_3}\approx1.769,
\)
whereas the largest player rate is
\(
    c^3\approx1.558.
\)
Thus, Remark~\ref{rem_full_activity_check} verifies the full-activity assumption without requiring a separate check of all nine coordinates.

\paragraph{From the zero-baseline benchmark to three zones.}
Next, we study the effect of the baseline loads.
The positive baselines can be introduced continuously while the equilibrium support changes discretely. Consider the same returns and resources as in the next example,
\[
    a=(10,9,8,3,2),
    \qquad
    r=(15,14,3,1),
\]
and scale the baseline vector according to
\[
    b(\eps)=\eps(1,1,1,1,1),
    \qquad
    0\le\eps\le1.
\]
At $\eps=0$, interpreted as the zero-baseline benchmark, Proposition~\ref{prop_zero_baseline_equilibrium} gives
\(
    x_i^{j,\ast}(0)
    =
    r^j{a_i}/{32},
\)
so every player invests in every project.
As $\eps$ increases, the equilibrium remains continuous, but coordinates hit zero at distinct boundary values. Solving the corresponding boundary equalities numerically gives
\[
    \eps_1\approx0.297496,
    \qquad
    \eps_2\approx0.554617,
    \qquad
    \eps_3\approx0.822807.
\]
The cutoff vector evolves according to the following phase diagram:
\begin{center}
\begin{tikzpicture}[x=11.5cm,y=0.75cm]
\fill[black!4] (0,0.35) rectangle (0.297496,1.15);
\fill[black!8] (0.297496,0.35) rectangle (0.554617,1.15);
\fill[black!12] (0.554617,0.35) rectangle (0.822807,1.15);
\fill[black!16] (0.822807,0.35) rectangle (1,1.15);
\draw (0,0.35) rectangle (1,1.15);
\draw (0.297496,0.35) -- (0.297496,1.15);
\draw (0.554617,0.35) -- (0.554617,1.15);
\draw (0.822807,0.35) -- (0.822807,1.15);
\node[font=\scriptsize] at (0.148748,0.75) {$(5,5,5,5)$};
\node[font=\scriptsize] at (0.426057,0.75) {$(5,5,5,4)$};
\node[font=\scriptsize] at (0.688712,0.75) {$(5,5,5,3)$};
\node[font=\scriptsize] at (0.911404,0.75) {$(5,5,4,3)$};
\node[font=\scriptsize,align=center] at (0.297496,1.48) {$x_5^{4,\ast}\downarrow0$};
\node[font=\scriptsize,align=center] at (0.554617,1.48) {$x_4^{4,\ast}\downarrow0$};
\node[font=\scriptsize,align=center] at (0.822807,1.48) {$x_5^{3,\ast}\downarrow0$};
\draw[->] (0,0) -- (1.04,0) node[right] {$\eps$};
\foreach \x/\lab in {0/0,0.297496/\eps_1,0.554617/\eps_2,0.822807/\eps_3,1/1}
    \draw (\x,0.06) -- (\x,-0.06) node[below=2pt,font=\scriptsize] {$\lab$};
\end{tikzpicture}
\end{center}
Thus, the positive baseline does more than perturb the proportional allocation: as its scale increases, the equilibrium passes through three support transitions and develops the nested zone structure proved in Proposition~\ref{thm_cutoff_zone_structure}. At $\eps=1$, this family coincides with the three-zone example considered next.

\paragraph{A three-zone equilibrium and Block Pandora.}
Consider
\[
    a=(10,9,8,3,2),\qquad
    b=(1,1,1,1,1),\qquad
    r=(15,14,3,1).
\]
As already mentioned above, such parameters lead to an equilibrium consisting of multiple zones. More precisely, a numerical simulation produces the equilibrium cutoffs and zones given by
\[
    (k^1,k^2,k^3,k^4)=(5,5,4,3),
    \qquad
    Z_1=\{1,2,3\},\
    Z_2=\{4\},\
    Z_3=\{5\}.
\]
The final frame of a numerical Block Pandora run is shown in Figure~\ref{fig_block_pandora_example}. Moreover, the boundary checks
\[
    \frac{a_4}{L_4}\approx0.801<c^4\approx0.827,
    \qquad
    \frac{a_5}{L_5}\approx0.755<c^3\approx0.776
\]
confirm the displayed support. In this implementation, the candidate support was enlarged from left to right until the boundary inequalities were satisfied.

\begin{figure}[H]
    \centering
    \includegraphics[width=0.90\textwidth]
        {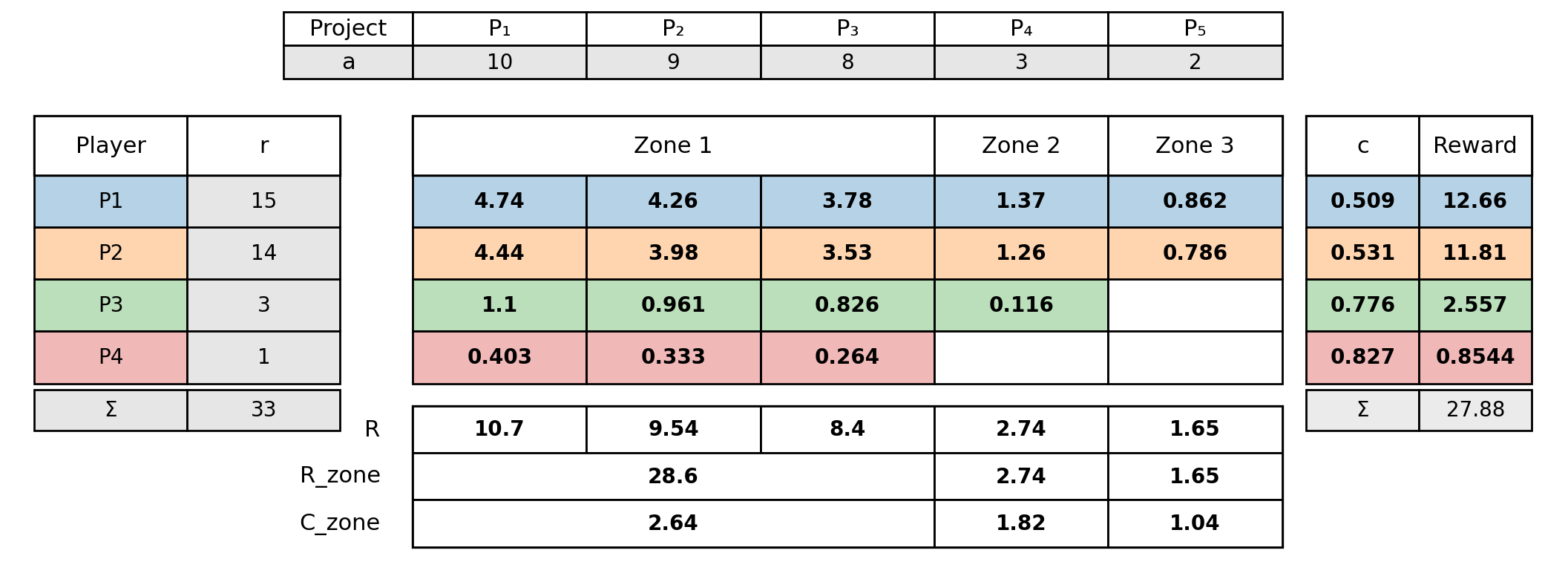}
    \caption{Final frame of a numerical Block Pandora run. Entries are equilibrium allocations; blank cells denote zero allocations, and the vertical divisions indicate the zones.}
    \label{fig_block_pandora_example}
\end{figure}


\paragraph{Best response may cycle.}
Finally, we give an example that motivates Section~\ref{section_algorithms} and shows that cyclic best response need not converge. Take
\[
    a=(9,4),\qquad b=(2,1),\qquad r=(70,1),
\]
so that $a$, $b$, and $a/b$ are all strictly decreasing.

Let $s_-$ be the first coordinate of player~$1$'s best response when player~$2$ chooses $(0,1)$, and $s_+$ the corresponding coordinate when player~$2$ chooses $(1,0)$. For the present parameters, player~$2$'s best response to $(s_-,70-s_-)$ is $(1,0)$, whereas her best response to $(s_+,70-s_+)$ is $(0,1)$. Hence, under cyclic exact best response in the order player~$2$, then player~$1$, one obtains the strict two-cycle
\[
    \begin{pmatrix}
        s_-&70-s_-\\
        0&1
    \end{pmatrix}
    \longmapsto
    \begin{pmatrix}
        s_+&70-s_+\\
        1&0
    \end{pmatrix}
    \longmapsto
    \begin{pmatrix}
        s_-&70-s_-\\
        0&1
    \end{pmatrix}.
\]
The two values can be obtained analytically from the one-dimensional best-response conditions:
\[
    s_-=212/5=42.4,
    \qquad
    s_+=\frac{213\sqrt3-6}{3\sqrt3+2}\approx50.433.
\]
The unique Nash equilibrium is instead
\[
    x^\ast\approx
    \begin{pmatrix}
        48.398&21.602\\
        0.762&0.238
    \end{pmatrix}.
\]
Writing $s^\ast\approx48.398$ for the first coordinate of player~$1$ at equilibrium, empirical analysis shows that $s^\ast$ is a repelling fixed point and that every initialization with $x_1^1\ne s^\ast$ is attracted to the displayed cycle. Combined with the rapid convergence observed in many other parameter regimes, this suggests a nontrivial stability geometry: characterizing the regions of convergence, periodicity, and more complicated behavior is an interesting open problem.

\section{Extensions and Open Problems}\label{section_open_problems}

We conclude the paper by discussing open problems and extensions.
The model studied in this paper is deliberately static and deterministic, uses the same project parameters for every player, and is built around a particular rational payoff function. These assumptions make possible the explicit one-zone formulas and the constructive solution, but each of them suggests a natural extension. We describe the directions that seem most immediate.

\paragraph{Player-specific project parameters and participation constraints.}
The most important extension, in our view, is to allow different players to have different returns from the same project. The simplest such model replaces $a_i$ by $a_i^j$, while a more general version also replaces $b_i$ by $b_i^j$. Player-specific parameters may represent different technologies, information, expertise, or other comparative advantages. Existence continues to follow from the standard theory of concave games under the natural positivity assumptions, but uniqueness and the ordered structure proved in this paper are no longer automatic. In particular, the players may rank the projects differently, so there need not be a common ordering by $a_i/b_i$, and nested cutoff sets may fail. A natural problem is to identify conditions, such as a common ranking or an appropriate single-crossing property, under which some version of the cutoff and zone structure survives.

A closely related extension is to restrict player $j$ to a prescribed set of projects $E^j\subseteq\{1,\ldots,n\}$. The strategy set then becomes
\(
    X^j(E^j)
    =
    \left\{
        x^j\in\mathbb R_+^n:
        \sum_{i=1}^n x_i^j=r^j,
        \quad x_i^j=0 \text{ for } i\notin E^j
    \right\}.
\)
This can be viewed formally as a zero-return specialization of the player-specific model, but the constrained formulation is cleaner and avoids allocations to unavailable projects. 

\paragraph{More general payoff functions.}
The rational function used in this paper is simple enough to be tractable, but already captures the two effects that motivated the model: diminishing returns and the crowding effect. It is therefore natural to consider separable payoffs of the form
\(
    F^j(x)
    =
    \sum_{i=1}^n
    \varphi_i^j\bigl(x_i^j,u_i^j\bigr),
\)
where $\varphi_i^j$ is increasing and concave in its first argument and decreasing in its second. The existence and equimarginal parts of the theory extend under standard concavity assumptions. Moreover, the proof of convergence of the projected marginal-utility algorithm uses only convexity of the strategy set together with Lipschitz continuity and strong anti-monotonicity of the pseudo-gradient. Thus, the same algorithm and proof apply to any class of payoff functions for which these two quantitative properties can be verified.
The more delicate questions concern the structural results. Our explicit formulas, and in particular the reduction of the fully active equilibrium to one scalar nonlinear equation, use the precise rational form of the payoff in an essential way. It would be useful to characterize the larger class of functions for which a comparable aggregate reduction remains possible. Even when no explicit formula survives, one may ask for conditions guaranteeing ordered active sets, monotonicity across players, and a structure-exploiting analogue of Block Pandora.

\paragraph{Stochastic project characteristics.}
Another broad direction is to introduce randomness into the project parameters. One may replace $a_i$ and $b_i$ by random variables $A_i$ and $B_i$, beginning, for example, with two-point or Bernoulli models and then passing to general distributions. There is, however, a simple distinction to keep in mind. If only $A_i$ is random, all players are risk neutral, the allocation is chosen before the realization, and all players have the same prior, then taking expectations merely replaces $a_i$ by $\mathbb E A_i$; this model reduces to the deterministic one. A genuinely stochastic problem arises when the congestion parameters or project availability are random, when players receive different signals, when investments can be revised after partial information is observed, or when the players use nonlinear risk criteria.
These variants lead naturally to Bayesian games. Among the first questions are whether pure or monotone Bayesian equilibria exist, whether posterior project rankings produce state-dependent cutoff structures, and whether the projected marginal-utility method can be adapted to expected or sample-based pseudo-gradients. 

\paragraph{Dynamic and multi-stage games.}
The present model has a single investment horizon. A repeated version with unchanged parameters and no intertemporal coupling would simply reproduce the same static equilibrium at every date, so an interesting dynamic model must allow the state to evolve. For example, the project parameters may be stochastic processes $A_i(t)$ and $B_i(t)$; players and projects may enter or leave; resources may be carried from one date to the next; and past returns may affect future budgets. One may also impose adjustment costs, irreversibility, or limits on how quickly an allocation can be changed.

\paragraph{Social welfare and mechanism design.}
Finally, one can replace individual payoff maximization by social-welfare maximization. Indeed, the noncooperative equilibrium need not maximize the sum of the players' payoffs. This direction already has a simple first answer. Let
\(
    y_i=\sum_{j=1}^m x_i^j
\)
and 
\(
    R=\sum_{j=1}^m r^j.
\)
Then social welfare depends only on the aggregate investments:
\[
    W(x)
    :=\sum_{j=1}^m F^j(x)
    =\sum_{i=1}^n \frac{a_i y_i}{b_i+y_i},
    \qquad
    \sum_{i=1}^n y_i=R.
\]
The planner's problem is therefore a strictly concave one-player allocation problem. Its unique aggregate solution is
\begin{equation}\label{eq_social_optimum_open_problems}
    y_i^{\mathrm{SO}}
    =
    \left(
        \sqrt{\frac{a_i b_i}{\lambda}}-b_i
    \right)_+,
\end{equation}
where $\lambda>0$ is uniquely chosen so that $\sum_i y_i^{\mathrm{SO}}=R$. Thus, ordering projects by $q_i=a_i/b_i$, the social optimum also has a cutoff structure. More explicitly, if the first $k$ projects are active, then
\[
    \sqrt{\lambda}
    =
    \frac{\sum_{i=1}^k\sqrt{a_i b_i}}
    {R+\sum_{i=1}^k b_i},
\]
with $q_k>\lambda\ge q_{k+1}$, where $q_{n+1}:=0$. The aggregate solution is unique, although its decomposition among players is not; one feasible socially optimal profile is
\(
    x_i^{j,\mathrm{SO}}
    =
    r^j y_i^{\mathrm{SO}} / R.
\)

The social optimum can also be decentralized by charging each player for the congestion externality she imposes on the others. Recall that $u_i^j=\sum_{\ell\ne j}x_i^\ell$. Define the fee charged to player $j$ on project $i$ by
\begin{equation}
    \tau_i^j(x)
    :=\frac{a_i u_i^j}{b_i+u_i^j}
      -\frac{a_i u_i^j}{b_i+u_i^j+x_i^j}
    =
    \frac{a_i u_i^j x_i^j}
    {(b_i+u_i^j)(b_i+u_i^j+x_i^j)}.
\end{equation}
This is exactly the reduction in the total payoff of the other players on project $i$ caused by player $j$'s investment. If the modified payoff is
\(
    \widetilde F^j(x)
    :=F^j(x)-\sum_{i=1}^n\tau_i^j(x),
\)
then, for fixed $x^{-j}$,
\[
    \widetilde F^j(x)
    =
    W(x)
    -
    \sum_{i=1}^n\frac{a_i u_i^j}{b_i+u_i^j},
\]
and the last term is independent of $x^j$. Hence, the modified game is an exact potential game with potential $W$, and its Nash equilibria coincide with the socially optimal profiles.

This calculation opens several mechanism-design questions. Since only the aggregate vector $y^{\mathrm{SO}}$ is unique, how should an organizer select a particular division of the projects among the players? Can the corrective fees be made budget balanced, individually rational, or robust to private project parameters? For the original unregulated game, it is also natural to study the efficiency loss of equilibrium, including sharp price-of-anarchy bounds and refinements that use the resource vector, project qualities, or the equilibrium zone structure.

The general question behind all these extensions is which parts of the theory are consequences of concavity and strategic congestion, and which rely on the special algebraic form of the present model. In particular, the robustness of uniqueness, nested activity, the aggregate scalar reduction, and monotone constructive algorithms appears to deserve further study.

\bigskip
\noindent
\textbf{Acknowledgments:} We are grateful to Jan Vecer for stimulating discussions and especially for drawing our attention to the social-welfare and mechanism design extension.

\bigskip
\noindent
\textbf{Declaration of generative AI and AI-assisted technologies in the manuscript preparation process:}
During the preparation of this work, the authors used OpenAI Chat GPT to assist with English-language editing, consistency checking, formatting, and submission preparation. After using this tool, the authors reviewed and edited the content as needed and take full responsibility for the content of the publication.

\begin{appendices}

\section{Auxiliary Proofs}\label{appendix_auxiliary_proofs}

\subsection{The zero-baseline benchmark}\label{appendix_zero_baseline_equilibrium}

\begin{Prop}\label{prop_zero_baseline_equilibrium}
Assume $m\ge2$, and consider the game
\[
    \widehat F^j(x)
    =
    \sum_{i=1}^n
    a_i\frac{x_i^j}{\sum_{\ell=1}^m x_i^\ell},
\]
where an unclaimed project yields zero payoff. Its unique Nash
equilibrium is
\[
    x_i^{j,\ast}
    =
    r^j\frac{a_i}{\sum_{h=1}^n a_h},
    \qquad
    i=1,\ldots,n,\quad j=1,\ldots,m.
\]
\end{Prop}

\begin{proof}
Put $A:=\sum_i a_i$ and $R:=\sum_j r^j$. At the displayed profile, the opponents of player $j$ invest $(R-r^j)a_i/A$ in project $i$. Consequently, player $j$'s marginal payoff on every project equals $A(R-r^j)/R^2$. Since her payoff is strictly concave in her own allocation, the equimarginal condition shows that this profile is a Nash equilibrium.

For uniqueness, let $x$ be any Nash equilibrium and write $L_i:=\sum_jx_i^j$. Every project must be claimed: otherwise, a player can transfer an arbitrarily small amount from a project she uses to an unclaimed project, obtain the payoff $a_i$ there, and incur an arbitrarily small loss on the former project. Moreover, no project can be used by only one player. Indeed, its sole user can transfer a small amount to a project used by another player, retain the full payoff $a_i$ on the first project, and obtain a strict gain on the second. Thus, $L_i-x_i^j>0$ for every $i,j$, and every player's best-response problem is smooth and strictly concave.

Let $c^j>0$ be the rate associated with player $j$'s resource constraint, and set $p_i:=a_i/L_i$. The KKT conditions give
\[
    \frac{x_i^j}{L_i}
    =
    \left(1-\frac{c^j}{p_i}\right)_+,
    \qquad
    1
    =
    \sum_{j=1}^m
    \left(1-\frac{c^j}{p_i}\right)_+
    =:h(p_i).
\]
The function $h$ is continuous and strictly increasing whenever it is positive, and satisfies $\lim_{p\to\infty}h(p)=m>1$. Hence $h(p)=1$ has a unique solution, so $p_i=p$ for all $i$. Because every player has a positive resource, every player is active on some project and therefore, since $p_i$ is independent of $i$, on every project. Thus, $x_i^j/L_i=\alpha^j$ is independent of $i$. The resource constraints imply $\alpha^j=r^j/R$, while $L_i=a_i/p$ and $\sum_iL_i=R$ imply $p=A/R$. Therefore, $x_i^j=r^ja_i/A$, proving uniqueness.
\end{proof}

\subsection{Proof of Lemma \ref{lem_concavity_antimonotonicity}}\label{appendix_lem_concavity_antimonotonicity}

\begin{proof}
We first prove the concavity. Fix a player $j$ and the opponents' profile $x^{-j}$. With $\widehat b_i:=b_i+u_i^j>0$, the contribution of project $i$ as a function of $t=x_i^j$ is $g_i(t)=a_it/(\widehat b_i+t)$, and
\[
    g_i''(t)=-\frac{2a_i\widehat b_i}{(\widehat b_i+t)^3}<0.
\]
Thus, $F^j(\cdot,x^{-j})$ is strictly concave as the sum of concave coordinatewise contributions.

For the quantitative estimate, take $x,y\in X$ and set $d:=x-y$, $d_i^j:=x_i^j-y_i^j$, $d_i:=(d_i^1,\ldots,d_i^m)$, and $\Delta_i:=\sum_jd_i^j$. Let $x(t):=y+td$ and $L_i(t):=L_i(x(t))$, so $b_i\le L_i(t)\le b_i+R$. Direct differentiation of $c_i^j(x(t))=a_i(L_i(t)-x_i^j(t))/L_i(t)^2$ gives
\[
    \frac{d}{dt}\sum_{j=1}^m d_i^j c_i^j(x(t))
    =\frac{a_i}{L_i(t)^3}\left[-L_i(t)\bigl(\Delta_i^2+\lVert d_i\rVert^2\bigr)+2\Delta_i\sum_{j=1}^m x_i^j(t)d_i^j\right].
\]
For every player $j$, we have
\(
    2\Delta_i d_i^j
    \le \Delta_i^2+(d_i^j)^2.
\)
Since $x_i^j(t)\ge0$, multiplying these inequalities by $x_i^j(t)$ and summing over $j$ yields
\begin{align*}
    2\Delta_i\sum_{j=1}^m x_i^j(t)d_i^j
    \le
    \sum_{j=1}^m
    x_i^j(t)\bigl(\Delta_i^2+(d_i^j)^2\bigr)
    \le
    \left(\sum_{j=1}^m x_i^j(t)\right)
    \bigl(\Delta_i^2+\lVert d_i\rVert^2\bigr)
    =
    \bigl(L_i(t)-b_i\bigr)
    \bigl(\Delta_i^2+\lVert d_i\rVert^2\bigr).
\end{align*}
Consequently,
\[
    \frac{d}{dt}\sum_{j=1}^m d_i^j c_i^j(x(t))
    \le-\frac{a_ib_i}{L_i(t)^3}\bigl(\Delta_i^2+\lVert d_i\rVert^2\bigr)
    \le-\frac{a_ib_i}{(b_i+R)^3}\lVert d_i\rVert^2.
\]
Integrating over $t\in[0,1]$, summing over $i$, and using the definition of $\mu$ yields
\[
    (x-y)\cdot\bigl(\Phi(x)-\Phi(y)\bigr)
    \le-\sum_{i=1}^n\frac{a_ib_i}{(b_i+R)^3}\lVert d_i\rVert^2
    \le-\mu\lVert x-y\rVert^2.
\]
This is \eqref{eq_strong_antimonotonicity} and completes the proof.
\end{proof}

\subsection{Proof of existence and uniqueness in Proposition~\ref{thm_existence_uniqueness}}\label{appendix_existence_uniqueness}

\begin{proof}
For each player $j$, the strategy set $X^j$ is nonempty, compact, and convex. The payoff $F^j$ is continuous on $X$ and is concave in the strategy of player $j$ when the other players' strategies are fixed. Thus, the usual Kakutani fixed-point argument for concave games gives existence of a Nash equilibrium: the best-response correspondence has nonempty compact convex values and a closed graph, so it has a fixed point.

As for uniqueness, suppose that $x$ and $y$ are two Nash equilibria. Since $x^j$ maximizes $F^j(\cdot,x^{-j})$ over the convex set $X^j$, the first-order optimality condition for a concave maximization problem gives
\(
    \nabla_{x^j}F^j(x)\cdot (y^j-x^j)\le0.
\)
Similarly,
\(
    \nabla_{x^j}F^j(y)\cdot (x^j-y^j)\le0.
\)
Adding these inequalities over all players $j$ yields
\(
    (x-y)\cdot\bigl(\Phi(x)-\Phi(y)\bigr)\ge0.
\)
If $x\ne y$, this contradicts the strong anti-monotonicity in Lemma~\ref{lem_concavity_antimonotonicity}. Hence $x=y$.
\end{proof}

\subsection{Proof of Lemma \ref{lem_basic_nonlinear_equation}}\label{appendix_lem_basic_nonlinear_equation}

\begin{proof}
We show that the function $G$ in \eqref{eq_basic_nonlinear_equation} is continuous and strictly decreasing on $(0,\infty)$, with
\(
    \lim_{C\downarrow0}G(C)=+\infty,
    \,
    \lim_{C\to\infty}G(C)=0.
\)
The limits and continuity are immediate from the explicit formula \eqref{Li_solution}.  To prove monotonicity, substitute $L_i$ by $G_i(C)$ in \eqref{eq_one_zone_load_quadratic} and differentiate the identity with respect to $C$. This gives
\[
    \frac{G_i(C)^2}{a_i}
    +\left(\frac{2CG_i(C)}{a_i}-(m-1)\right)G_i'(C)=0.
\]
Dividing \eqref{eq_one_zone_load_quadratic} by $G_i(C)>0$ (again, after substituting $L_i$ by $G_i(C)$), we get
\[
    \frac{CG_i(C)}{a_i}=(m-1)+\frac{b_i}{G_i(C)}, 
    \quad \text{ thus also }\quad 
    \frac{2CG_i(C)}{a_i}-(m-1)
    =(m-1)+\frac{2b_i}{G_i(C)}>0.
\]
It follows that $G_i'(C)<0$ for each $i$, and hence $G'(C)<0$.  Thus, $G$ maps $(0,\infty)$ continuously and strictly decreasingly onto $(0,\infty)$, so \eqref{eq_basic_nonlinear_equation} has a unique solution.  Since increasing $R$ increases the left-hand side $R+B$, while $G$ is decreasing in $C$, the corresponding solution $C$ decreases.
\end{proof}

\subsection{Proof of Lemma \ref{lem_pseudogradient_lipschitz}}\label{appendix_lem_pseudogradient_lipschitz}

\begin{proof}
For each project $i$, write
\(
    x_i:=(x_i^1,\ldots,x_i^m),
    \,
    c_i(x):=(c_i^1(x),\ldots,c_i^m(x)),
\)
and let $J_i(x)$ be the Jacobian of $c_i$ with respect to $x_i$. Differentiating
\[
    c_i^j(x)=\frac{a_i\bigl(L_i(x)-x_i^j\bigr)}{L_i(x)^2}
\]
gives
\[
    [J_i(x)]_{j\ell}
    =
    \begin{cases}
        \displaystyle
        -\frac{2a_i\bigl(L_i(x)-x_i^j\bigr)}{L_i(x)^3},
        & \ell=j,\\[7pt]
        \displaystyle
        \frac{a_i\bigl(2x_i^j-L_i(x)\bigr)}{L_i(x)^3},
        & \ell\ne j.
    \end{cases}
\]
Since $L_i(x)\ge b_i$ and $0\le x_i^j\le L_i(x)$, the diagonal entries have absolute value at most $2a_i/b_i^2$, and the off-diagonal entries have absolute value at most $a_i/b_i^2$. Hence both the maximum row sum and the maximum column sum of $J_i(x)$ are bounded by $(m+1)a_i/b_i^2$, and therefore
\[
    \lVert J_i(x)\rVert_2
    \le
    \sqrt{\lVert J_i(x)\rVert_1\lVert J_i(x)\rVert_\infty}
    \le
    (m+1)\frac{a_i}{b_i^2}
    \le K_\Phi.
\]

Now let $x(t):=y+t(x-y)$. Since $c_i$ depends only on the coordinates of project $i$, the fundamental theorem of calculus gives
\(
    c_i(x)-c_i(y)
    =
    \int_0^1
    J_i(x(t))\bigl(x_i-y_i\bigr)\,dt,
\)
and hence
\(
    \lVert c_i(x)-c_i(y)\rVert
    \le
    K_\Phi\lVert x_i-y_i\rVert.
\)
Finally, regrouping the coordinates of $\Phi$ projectwise,
\(
    \lVert\Phi(x)-\Phi(y)\rVert^2
    =
    \sum_{i=1}^n
    \lVert c_i(x)-c_i(y)\rVert^2
    \le
    K_\Phi^2
    \sum_{i=1}^n
    \lVert x_i-y_i\rVert^2
    =
    K_\Phi^2\lVert x-y\rVert^2.
\)
Taking square roots proves \eqref{eq_pseudogradient_lipschitz}.
\end{proof}

\end{appendices}

\printbibliography

\end{document}